\documentclass[twoside]{IEEEtran}
\usepackage{cite}
\usepackage{graphicx}
\usepackage{psfrag}
\usepackage{url}
\usepackage{amsmath}
\usepackage{array}
\usepackage{amssymb}
\usepackage{amsfonts}
\usepackage{epstopdf}
\usepackage{mathrsfs}

\newtheorem{lemma}{Lemma}

\newtheorem{remark}{Remark}
\newtheorem{definition}{Definition}
\newtheorem{theorem}{Theorem}
\newtheorem{example}{Example}

\usepackage{subfigure}
\usepackage{amsbsy}
\usepackage{caption}
\newcommand{\lt}{\underset{n \rightarrow \infty}{\lim}}

\title{Learning Immune-Defectives Graph through Group Tests}
\author{
\authorblockN{Abhinav Ganesan, Sidharth Jaggi, and Venkatesh Saligrama, \IEEEmembership{Senior Member, IEEE}}
\thanks{This work was done when A. Ganesan was a Post Doctoral Fellow at the Chinese University of Hong Kong, Hong Kong SAR (e-mail: abhinavg.rsh@gmail.com). S. Jaggi is with the Chinese University of Hong Kong, Hong Kong SAR (e-mail: jaggi@ie.cuhk.edu.hk). V. Saligrama is with Boston University, Boston, USA (e-mail: srv@bu.edu).

A part of the content of this paper has been presented at the 2015 IEEE International Symposium on Information Theory (ISIT).}
}
\begin{document}

\maketitle
\thispagestyle{empty}	

\begin{abstract}
This paper deals with an abstraction of a unified problem of drug discovery and pathogen identification. Pathogen identification involves identification of disease-causing biomolecules. Drug discovery involves finding chemical compounds, called lead compounds, that bind to pathogenic proteins and eventually inhibit the function of the protein. In this paper, the lead compounds are abstracted as inhibitors, pathogenic proteins as defectives, and the mixture of ``ineffective'' chemical compounds and non-pathogenic proteins as normal items. A defective could be immune to the presence of an inhibitor in a test. So, a test containing a defective is positive iff it does not contain its ``associated'' inhibitor. The goal of this paper is to identify the defectives, inhibitors, and their ``associations'' with high probability, or in other words, learn the Immune Defectives Graph (IDG) efficiently through group tests.  We propose a probabilistic non-adaptive pooling design, a probabilistic two-stage adaptive pooling design and decoding algorithms for learning the IDG. For the two-stage adaptive-pooling design, we show that the sample complexity of the number of tests required to guarantee recovery of the inhibitors, defectives, and their associations with high probability, i.e., the upper bound, exceeds the proposed lower bound by a logarithmic multiplicative factor in the number of items. To be precise, lower and upper bounds of $\Omega\left((r+d)\log n + rd\right)$ and $O\left(rd \log n\right)$ tests respectively are identified for classifying $r$ inhibitors and $d$ defectives amongst $n$ items, and their associations. For the non-adaptive pooling design, we show that the upper bound (given by $O((r+d)^2 \log n)$ tests) exceeds the proposed lower bound (given by $\min \left\{\Omega\left((r+d)\log n + rd\right),\Omega\left(\frac{r^2}{\log r}\log n\right),\Omega\left(d^2\right)\right\}$ tests) by at most a logarithmic multiplicative factor in the number of items.
\end{abstract}	

\section{Introduction} \label{sec1}
Preliminary stages of drug discovery involve finding `blocker' or `lead' compounds that bind to a biomolecular target, which is a disease causing pathogenic protein, in order to inhibit the function of the protein. Such compounds are later used to produce new drugs. These lead compounds have to be identified amidst billions of chemical compounds \cite{XTSY2001,Cli2014}, and hence drug discovery is a tedious process. A complementary problem involves identifying pathogenic proteins amidst non-pathogenic ones, both of which are structurally identical in some respects. For instance, out of five known species of ebolavirus, only four of them are pathogenic to humans (see p. $5$ in \cite{Cli2014}) and a similar example can be found in arenavirus \cite{XLL2014}. Some of these pathogenic proteins might share a common inhibitory mechanism against a lead compound which serves to distinguish them from the non-pathogenic ones \cite{XLL2014}. So, finding potential pathogenic proteins amidst a large collection of biomolecules by testing them against known inhibitory compounds is a problem complementary to the problem of lead compound discovery. The lead compounds can be abstracted as inhibitor items, the pathogenic proteins as defective items, and the others as normal items. Now, the above problems can be combined to be viewed as an inhibitor-defective classification problem on the mixture of pathogenic and non-pathogenic proteins, and billions of chemical compounds. This unifies the process of finding both the pathogenic proteins and the lead compounds. An efficient means of solving this problem could  potentially be applied in high-throughput screening for drugs and pathogens or computer-assisted drug and pathogen identification. A natural consideration is that, while some pathogenic proteins might be inhibited by some lead compounds, other pathogenic proteins might be immune to some of these lead compounds present in the mixture of items. In other words, each defective item is possibly immune to the presence of some inhibitor items so that its expression cannot be prevented by the presence of those inhibitors when tested together. By definition, an inhibitor inhibits at least one defective. Learning this inhibitor-defective interaction as well as classifying the inhibitors and defectives efficiently through group testing is presented this work.

A representation of this model, which we refer to as the Immune-Defectives Graph (IDG) model, is given in Fig. \ref{fig-IDG}. The presence of a directed edge between a pair of vertices $\left(w_{i_{k_1}},w_{j_{k_2}}\right)$ represents the inhibition of the defective $w_{j_{k_2}}$ by the inhibitor $w_{i_{k_1}}$ and the absence of a directed edge between a pair of vertices $\left(w_{i_{k_1}},w_{j_{k_2}}\right)$ indicates that the inhibitor $w_{i_{k_1}}$ does not affect the expression of the defective $w_{j_{k_2}}$ when tested together. A formal presentation of the IDG model and the goals of this paper appear in the next section.
\begin{figure}[htbp]  
\centering
\includegraphics[totalheight=2.8in,width=3.8in]{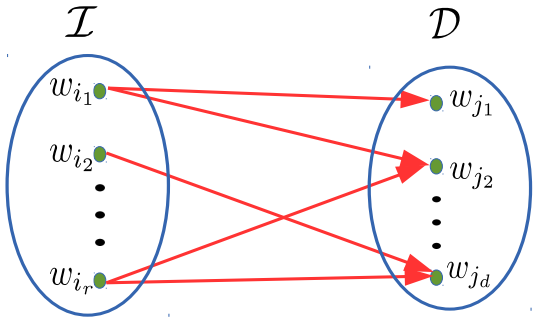}
\vspace{-1cm}
\caption{A representation of the IDG Model, where ${\cal I}$ represents the set of inhibitors and ${\cal D}$ represents the set of defectives.}
\label{fig-IDG}
\end{figure}

\begin{example}\label{eg-bipartite}
An instance of the IDG model is given in Fig. \ref{fig-Eg1}. In this example, the outcome of a test is positive iff a defective $w_{j_{k_2}}$, for some $k_2$, is present in the test and its associated inhibitor $w_{i_{k_2}}$ does not appear in the test. Observe that if the item-pair $\left(w_{i_{k_1}},w_{j_{k_2}}\right)$, for $k_1 \neq k_2$, appears in a test and $w_{i_{k_2}}$ does not appear in the test, then the outcome is positive. Also, if the item-pair $\left(w_{i_{k_2}},w_{j_{k_2}}\right)$ appears in a test and if $w_{j_{k'_2}}$ also appears in the test but not $w_{i_{k'_2}}$, then the test outcome is positive. But if the appearance of every defective $w_{j_{k'_2}}$ in a test is compensated by the appearance of its associated inhibitor $w_{i_{k'_2}}$ in the test, then the test outcome is negative. The outcome of a test is also negative when none of the defectives appear in a test.

\begin{figure}[htbp]  
\centering \hspace{-1cm}
\includegraphics[totalheight=2.8in,width=3.8in]{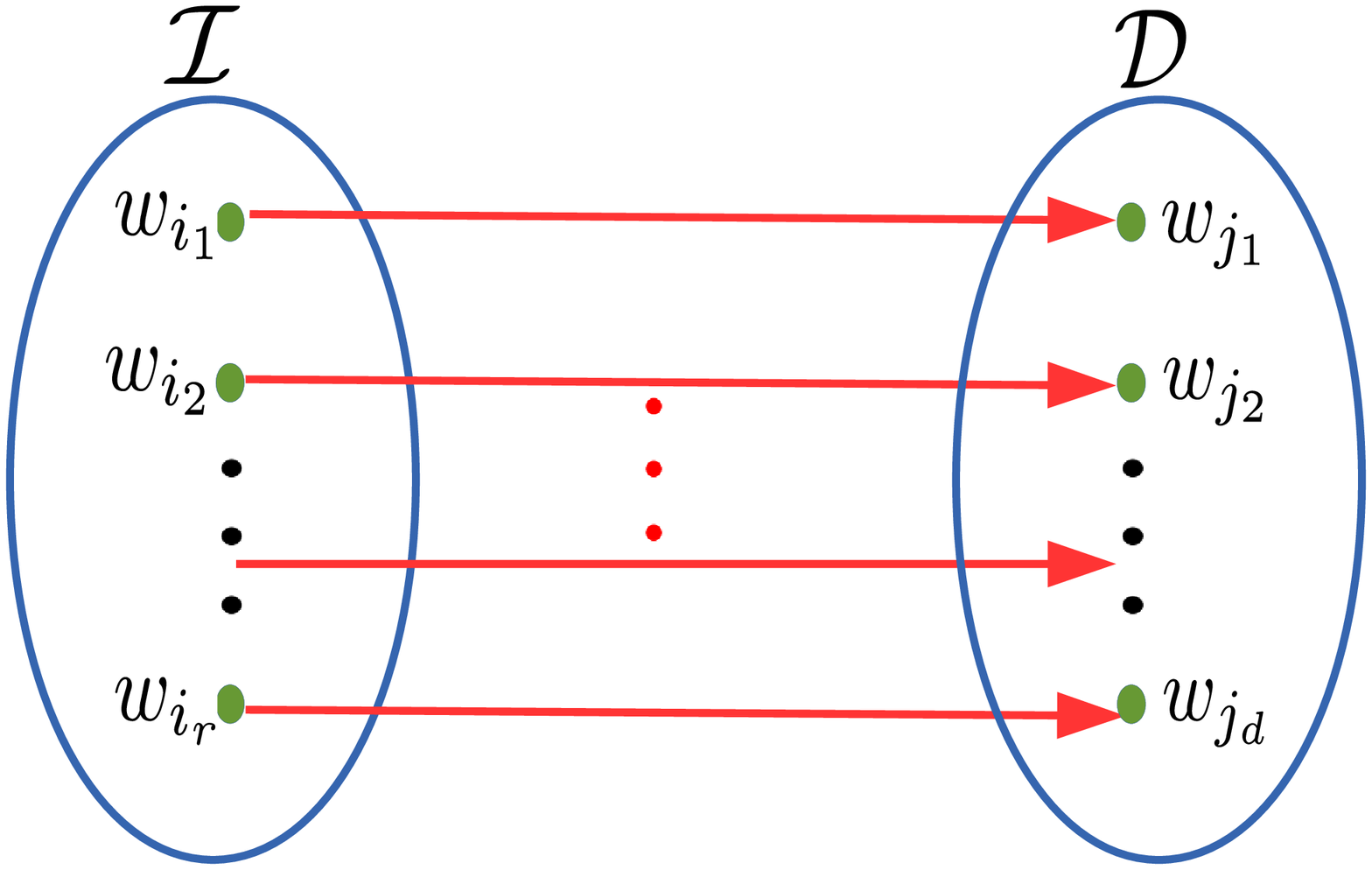}
\vspace{-1cm}
\caption{An example for the IDG Model where each defective is associated with a distinct inhibitor so that $r=d$.}
\label{fig-Eg1}
\end{figure}
\end{example}

The IDG model can also be viewed as a generalization of the $1$-inhibitor model introduced by Farach et al. in \cite{FKKM1997}. This model was motivated by errors in blood testing where blocker compounds (i.e., inhibitors) block the expression of defectives in a test  \cite{PhS1994}. This is also motivated by drug discovery applications where the inhibitors are actually desirable items that inhibit the pathogens \cite{LOGY1997}. In the $1$-inhibitor model, a test outcome is positive iff there is at least one defective and no inhibitors in the test. So, the presence of a single inhibitor is sufficient to ensure that the test outcome is negative.

Efficient testing involves pooling different items together in every test so that the number of tests can be minimized  \cite{Dor1943}. Such a testing methodology is called group testing. The pooling methodology can be of two kinds, namely non-adaptive and adaptive pooling designs. In non-adaptive pooling designs, any pool constructed for testing is independent of the previous test outcomes, while in adaptive pooling designs, some constructed pools might depend on the previous test outcomes. A $k$-stage adaptive pooling design is comprised of pool construction and testing in $k$-stages, where the pools constructed for (non-adaptive) testing in the $k^{\text{th}}$ stage depend on the outcomes in the previous stages. While adaptive group testing requires lesser number of tests than non-adaptive group testing, the latter inherently supports parallel testing of multiple pools. Thus, non-adaptive group testing is more economical (because it allows for automation) as well as saves time (because the pools can be prepared all at once) which are of concern in library screening applications \cite{BBTK-Springer1996}. The $1$-inhibitor model has been extensively studied, and several adaptive and non-adaptive pooling designs for classification of the inhibitors and the defectives are known (refer, \cite{DeV1998,DMTV2001,ADB2008,CCF2010}). A detailed survey of known non-adaptive and adaptive pooling designs for the $1$-inhibitor model is given in \cite{GEJS2014}. The best (in terms of number of tests) known non-adaptive pooling design that guarantees high probability classification of the inhibitors and defectives is proposed in \cite{GEJS2014}. The non-adaptive pooling design proposed in \cite{GEJS2014} requires $O(d \log n)$ tests in the $r=O(d)$ regime and $O\left(\frac{r^2}{d} \log n\right)$ tests in the $d=o(r)$ regime to guarantee classification of both the inhibitors and defectives with high probability\footnote{The number of inhibitors, defectives and normal items are denoted by $r$, $d$, and $n-d-r$ respectively.}. In the small inhibitor, i.e.,  $r=O(d)$ regime, the upper bound on the number of tests matches with the lower bound while in the large inhibitor, i.e., $d=o(r)$ regime, the upper bound exceeds the lower bound of $O\left(\frac{r^2}{d \log \frac{r}{d}} \log n\right)$ by a $\log {\frac{r}{d}}$ multiplicative factor. Nonetheless, the $1$-inhibitor model constrains that every inhibitor must inhibit every defective, which is likely to be a tight requirement in practice. So, the IDG model is a more practical variant of the $1$-inhibitor model.

A formal presentation of the IDG model and the goals of this paper are given in the next section.

{\it Notations:} The Bernoulli distribution with parameter $p$ is denoted by ${\cal B}(p)$, where $p$ denotes the probability of the Bernoulli random variable taking a value of one. The set of binary numbers is denoted by $\mathbb{B}$. Matrices are indicated by boldface uppercase letters and vectors by boldface lowercase letters. The row-$i$, column-$j$ entry of a matrix $\mathbf{M}$ is denoted by $\mathbf{M}(i,j)$, and the coordinate-$i$ of a vector $\mathbf{y}$ is denoted by $\mathbf{y}(i)$. All the logarithms in this paper are taken to the base two. The probability of an event $\cal E$ is denoted by $\Pr \{\cal E\}$. The notation $f(n) \approx g(n)$ represents approximation of a function $f(n)$ by $g(n)$. Mathematically, the approximation denotes that for every $\epsilon>0$, there exists $n_0$ such that for all $n>n_0$, $1-\epsilon<\frac{|f(n)|}{|g(n)|}<1+\epsilon$. 

\section{The IDG Model} \label{sec2}
Consider a set of items ${\cal W}$ indexed as $w_1,\cdots,w_n$ comprised of $r$ inhibitors, $d$ defectives, and $n-r-d$ normal items. It is assumed throughout the paper that $r,d=o(n)$. 

\begin{definition}
An item pair $(w_i,w_j)$, for $i \neq j$, is said to be {\em associated} when the inhibitor $w_i$ inhibits the expression of the defective $w_j$. An item pair $(w_i,w_j)$, for $i \neq j$, is said to be {\em non-associated} if either the inhibitor $w_i$ does not inhibit the expression of the defective $w_j$ or if $w_i$ is not an inhibitor or if $w_j$ is not a defective.
\end{definition}
In general, the mention of an item pair $(w_i,w_j)$ need not mean that $w_i$ is an inhibitor and $w_j$ is a defective. This is understood from the context. 

\begin{definition}
 An {\it association graph} is a left to right directed bipartite graph $\pmb{\cal B}=({\cal I},{\cal D},{\cal E})$, where the set of vertices (on the left hand side) ${\cal I}=\{w_{i_1},w_{i_2}, \cdots, w_{i_r}\} \subset {\cal W}$ denotes the set of inhibitors, the set of vertices (on the right hand side) ${\cal D}=\{w_{j_1},w_{j_2},\cdots,w_{j_d}\} \subset {\cal W}$ denotes the set of defectives, and ${\cal E}$ is a collection of directed edges from ${\cal I}$ to ${\cal D}$. A directed edge $e=(w_{i},w_{k}) \in {\cal E}$, for ${i} \in \{i_1,\cdots,i_r\}$, ${j} \in \{j_1,\cdots,j_d\}$, denotes that the inhibitor $w_{i}$ inhibits the expression of the defective $w_{k}$.
\end{definition}

We refer to ${\cal E}({\cal I},{\cal D})$ {\em conditioned on the sets} $({\cal I},{\cal D})$ to be the {\it association pattern} on $({\cal I},{\cal D})$.

A pooling design is denoted by a test matrix $\mathbf{M}\in \mathbb{B}^{T \times n}$, where the $j^{\text{th}}$ item appears in the $i^{\text{th}}$ test iff $\mathbf{M}(i,j)=1$. {\em A test outcome is positive iff the test contains at least one defective without any of its associated inhibitors}. A positive outcome is denoted by one and a negative outcome by zero.

It is assumed throughout the paper that the defectives are not mutually obscuring, i.e., a defective does not function as an inhibitor for some other defective. In other words, the set of inhibitors ${\cal I}$ and the set of defectives ${\cal D}$ are disjoint.

The goal of this paper is to identify the association graph, or in informal terms, learn the IDG. Thus, the objectives are two-fold as represented by Fig. \ref{fig-obj}.
\begin{enumerate}
 \item Identify all the defectives.
 \item Identify all the inhibitors and also their association pattern with the defectives. 
\end{enumerate}
\begin{figure}[htbp]  \vspace{-0.8cm}
\centering 
\hspace*{-1cm}
\includegraphics[totalheight=3.1in,width=4in]{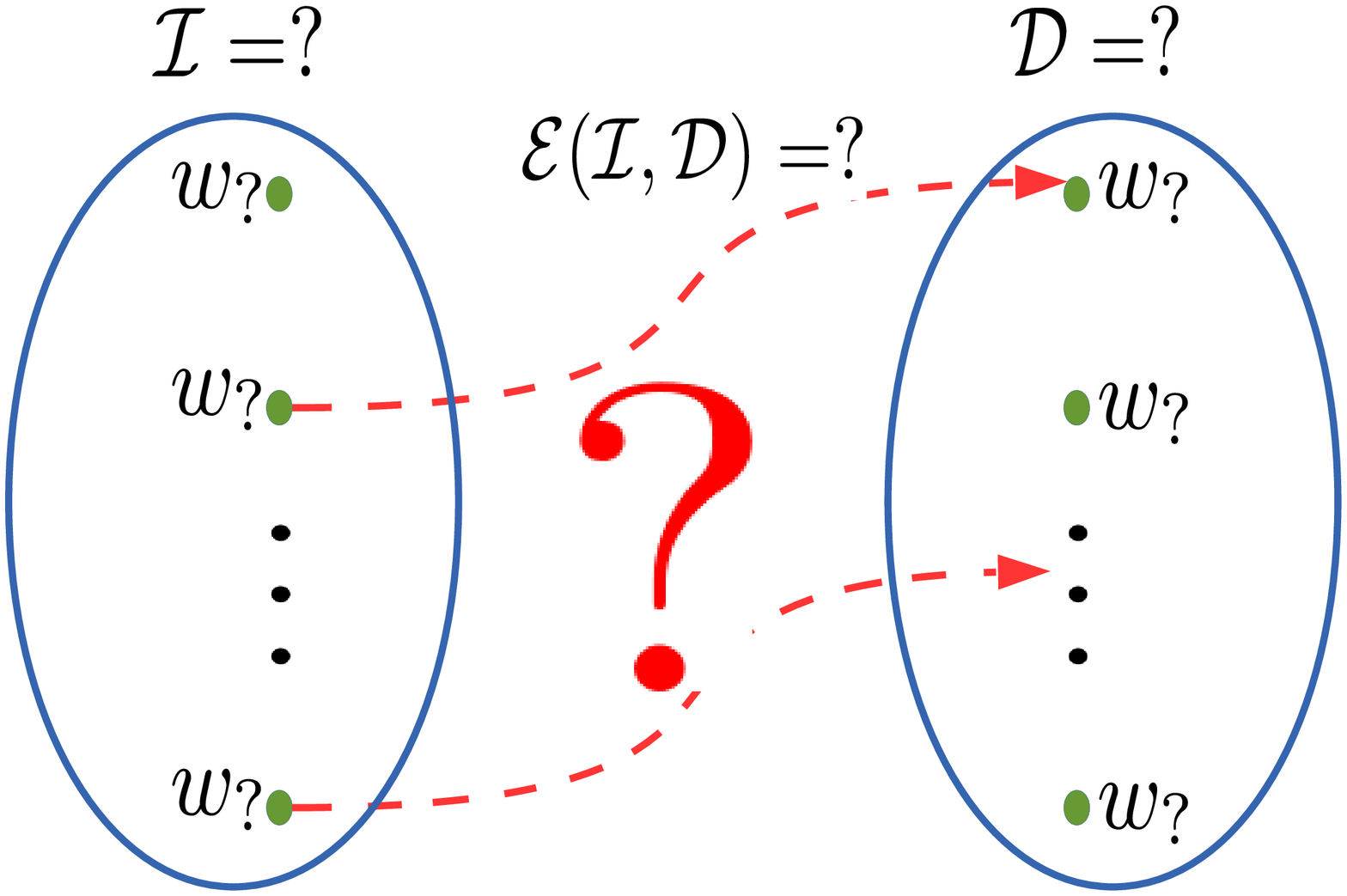}
\vspace{-1.5cm}
\caption{Here, the presence of a directed arrow represents an association between an inhibitor and a defective. The problem statement is to identify the set of inhibitors ${\cal I}$, defectives ${\cal D}$ and the association pattern ${\cal E}(\cal I,\cal D)$.}
\label{fig-obj}
\end{figure}
This problem is further mathematically formulated as follows. Denote the actual set of inhibitors, normal items, and defectives  by ${\cal I}$, ${\cal N}$, and ${\cal D}$ respectively so that ${\cal I}\cup{\cal N}\cup{\cal D}={\cal W}$. The actual association pattern between the actual inhibitor and defective sets is represented by ${\cal E}({\cal I},{\cal D})$. Let  $\hat{\cal I}$, $\hat{\cal N}$, $\hat{\cal D}$, and $\hat{\cal E}(\hat{\cal I},\hat{\cal D})$ denote the declared set of inhibitors, normal items, defectives, and declared association pattern between $(\hat{\cal I},\hat{\cal D})$ respectively. The target is to meet the following error metric.
\begin{align} \label{eqn-error_metric}
 \underset{{\cal I}, {\cal D},{\cal E}\left({\cal I}, {\cal D}\right)}{\max}  \Pr\left\{\left(\hat{\cal I}, \hat{\cal D},\hat{\cal E}\left(\hat{\cal I}, \hat{\cal D}\right)\right)\neq \left({\cal I}, {\cal D},{\cal E}({\cal I},{\cal D})\right)\right\} \leq cn^{-\delta},
\end{align} for some constants $c,\delta>0$.
We propose pooling designs and decoding algorithms, and lower bounds on the number of tests required to satisfy the above error metric. It is assumed that the defective and the inhibitor sets are distributed uniformly across the items, i.e., the probability that any given set of $r+d$ items constitutes all the defectives and inhibitors is given by $\frac{1}{{n \choose d}{n-d \choose r}}$. It is also assumed that the association pattern ${\cal E}({\cal I},{\cal D})$ is uniformly distributed over all possible association patterns on $({\cal I},{\cal D})$. 

We consider two variants of the IDG model. The first being the case where the maximum number of inhibitors that can inhibit any defective, given by $I_{max}$, is known. We refer to this model as the IDG with side information (IDG-WSI) model. For example, Fig. \ref{fig-Eg1} represents a case where $I_{max}=1$. While it is known that $I_{max}=1$, it is unknown which among the items $w_1,\cdots,w_n$ represent which inhibitors and defectives. For a given value of $(r,d)$, not all positive integer values of $I_{max}\leq r$ might be feasible. For instance, if $(r,d)=(3,2)$,  then $I_{max}=1$ is not feasible because, by definition, each inhibitor is associated with at least one defective. So, in the IDG-WSI model, we assume that the given value of $I_{max}$ is feasible for the $(r,d)$ tuple. In particular, if \mbox{$(c-1)d< r \leq cd$} for some integer $c \geq 1$, then $I_{max} \geq c$. This immediately follows from the fact that each inhibitor must be associated with at least one defective.

The other variant of the IDG model we consider in this paper is the case where there is no side information about the inhibitor-defective associations, which means that each defective can be inhibited by as many as $r$ inhibitors. We refer to this model as the IDG-No Side Information (IDG-NSI) model. For both the models, the goals (as stated in the beginning of this section) are the same.

The contributions of this paper for the IDG models are summarized below.

\begin{itemize}
\item The sample complexity of the number of tests sufficient to recover the association graph while satisfying the error metric (\ref{eqn-error_metric}) using the proposed 
\begin{itemize}
 \item non-adaptive pooling design is given by $T_{NA}=O\left((r+d)^2 \log n\right)$ and $T_{NA}=O\left((I_{\max}+d)^2 \log n\right)$ tests for the IDG-NSI and IDG-WSI models respectively (Theorem \ref{thm-NA_UB}, Section \ref{sec3}). 
 \item two-stage adaptive pooling design is given by $T_{A}=O\left(rd \log n\right)$ and $T_{A}=O\left(I_{\max}d \log n\right)$ tests for the IDG-NSI and IDG-WSI models respectively (Theorem \ref{thm-A_UB}, Section \ref{sec3}). 
\end{itemize}
\item In Section \ref{sec4} (Theorem \ref{thm-LB_GTI_ID_NSI} and Theorem \ref{thm-LB_GTI_ID_WSI}), lower bounds of 
\begin{align*}
&\max\left\{\Omega\left((r+d)\log n+rd\right), \Omega\left(\frac{r^2}{\log r} \log n\right),\Omega(d^2)\right\},\\ 
&\max\left\{\Omega\left((r+d)\log n+I_{max}d\right), \Omega\left(\frac{I_{max}^2}{\log I_{max}} \log n\right),\right.\\&\hspace{0.9cm}\left. \Omega(d^2)\right\}
\end{align*}are obtained for non-adaptive pooling designs for the IDG-NSI and IDG-WSI models respectively. The first lower bounds for both the models are valid for adaptive pooling designs also. The third lower bound for the IDG-WSI model is valid under some mild restrictions on $I_{max}$ and $r$, the details of which are given in Theorem \ref{thm-LB_GTI_ID_WSI}. 
\end{itemize}

The pooling design matrix $\mathbf{M}$ constructed in this paper use carefully chosen ``random matrices'', i.e., the entries of the matrices are chosen independently from a suitable Bernoulli distribution. Such matrices are known to permit ease of analysis \cite{book-HD}. Notwithstanding the simplicity of the pooling design construction, figuring out a good decoding algorithm with a reasonable computational complexity and good lower bounds, especially for non-adaptive pooling designs, is a challenging task. The goodness of the pooling design, decoding algorithm tuple and the proposed lower bounds is measured in terms of the closeness of the upper bounds to the lower bounds on the number of tests. For non-adaptive pooling designs, this can be observed from Table \ref{tab-UB_LB}.  For the proposed adaptive pooling design, the upper bound exceeds the lower bound by at most a $\log n$ multiplicative factor for both IDG-NSI and IDG-WSI models. Also, the proposed decoding algorithms have a computational complexity of $O(n T_{NA})$ and $O(n T_{A})$ time units for the non-adaptive and adaptive pooling designs, respectively. This intuitively means that an item is ``processed'' at most a constant number of times per test.

\begin{table*}
\captionsetup{font=small}
\centering
\caption{Necessary and sufficient number of tests for various regimes of the number of inhibitors, defectives, and $I_{max}$ are given. In the large inhibitor regime, i.e., $d=O(r)$ for the IDG-NSI model and $d=O(I_{max})$ for the IDG-WSI model, the upper bounds exceed the lower bounds by multiplicative factors of $\log r$ and $\log I_{max}$ for the IDG-NSI and IDG-WSI models respectively. In the small inhibitor regime, i.e., $r=o(d)$ for the IDG-NSI model and $I_{max}=o(d)$ for the IDG-WSI model, the upper bounds exceed the lower bounds by multiplicative factors of $\log n$ for both IDG-NSI and IDG-WSI models.} \vspace{0.5cm}
 \begin{tabular}{|c|c|c|}
 \hline
Model & $d=O(r), d=O(I_{max})$ (large inhibitor regime) & $r=o(d), I_{max}=o(d)$ (small inhibitor regime)\\
\hline
IDG-WSI & Upper Bound: $O\left(r^2 \log n\right)$ & Upper Bound: $O(d^2 \log n)$\\
        & Lower Bound: $\Omega\left(\frac{r^2}{\log r} \log n\right)$ & Lower Bound: $\Omega(d^2)$\\
\hline
IDG-NSI & Upper Bound: $O\left(I_{max}^2 \log n\right)$ & Upper Bound: $O(d^2 \log n)$\\
        & Lower Bound: $\Omega\left(\frac{I_{max}^2}{\log I_{max}} \log n\right)$ & Lower Bound: $\Omega(d^2)$\\
\hline
 \end{tabular}
 \label{tab-UB_LB}
\end{table*}

{\it Extension of the results on the upper and lower bounds on the number of tests to the case where only upper bounds on the number of inhibitors (given by $R$) and defectives (given by $D$) are known instead of their exact numbers is straightforward. The target error metric in (\ref{eqn-error_metric}) is re-formulated as maximum error probability criterion over all combinations of number of inhibitors and defectives. The results for this case follow by replacing $r$ by $R$ and $d$ by $D$ in the upper and lower bounds on the number of tests.}

There are various generalizations of the $1$-inhibitor model considered in the literature. These models are summarized in the following sub-section to show that the model considered in this paper, to the best of our knowledge, has not been studied in the literature.

\subsection{Prior Works}
The $1$-inhibitor model can be generalized in various directions, mostly influenced by generalizations of the classical group testing model. The various generalizations are listed below and briefly described. Though none of these generalizations include the model studied in this paper, it is worthwhile to understand the differences between these models and the IDG model.

A generalization of the $1$-inhibitor model, namely $k$-inhibitor model was introduced in \cite{DeV2003_k_inh}. In the $k$-inhibitor model, an outcome is positive iff a test contains at least one defective and no more than $k-1$ inhibitors. So, the number of inhibitors must be no less than a certain threshold $k$ to cancel the effect of any defective. This model is different from the model introduced in this paper because, in the IDG model, a single associated inhibitor is enough to cancel the effect of a defective. Further, none of the inhibitors might be able to cancel the effect of a defective because the defective might not be associated with any inhibitor. A model loosely related with the $1$-inhibitor model, namely mutually obscuring defectives model was introduced in \cite{Dam1998_MOD}. Here, it was assumed that multiple defectives could cancel the effect of each other, and hence the outcome of a  test containing multiple defectives could be negative. Thus, a defective can also function as a inhibitor. However, in this paper, the sets of defectives and inhibitors are assumed to be disjoint. The threshold (classical) group testing model is where a test outcome is positive if the test contains at least $u$ defectives, negative if it contains no more than $l$ defectives and arbitrarily positive or negative otherwise \cite{Dam2006_TGT}. This model was combined with the $k$-inhibitor model and non-adaptive pooling designs for the resulting model was proposed in \cite{HTZWG}. 

A non-adaptive pooling design for the general inhibitor model was proposed in \cite{HwC2007_GIM}. Here, the goal was to identify all the defectives with no prior assumption on the cancellation effect of the inhibitors on the defectives, i.e, the underlying unknown inhibitor model could be a $1$-inhibitor, $k$-inhibitor model, or even the ID model introduced in this paper. However, the difference from our work is that, we aim to identify the association graph or, in other words, the cancellation effect of the inhibitors also apart from identification of the defectives. But this cancellation effect does not include the $k$-inhibitor model cancellation effect as noted earlier. Group testing on complex model was introduced in \cite{Tor1999_CM}. In the complex model, a test outcome is positive iff the test contains at least one of the defective sets. So, here the notion of defectives items is generalized to sets of defective items called defective sets. This complex model was combined with the general inhibitor model and non-adaptive pooling designs for identification of defectives was proposed in \cite{CCH2011_ICM}. Our work is different from \cite{CCH2011_ICM} for the same reasons as stated for \cite{HwC2007_GIM}. Group testing on bipartite graphs was proposed in \cite{LTLW2005_BGT} as a special case of the complex model. Here, the left hand side of the bipartite graph represents the bait proteins and the right hand side represents the prey proteins. It is known a priori which items are baits and which ones are preys. The edges in the bipartite graph represent associations between the baits and preys. A test outcome is positive iff the test contains associated items and the goal was to identify these associations. Clearly, this model is different from the IDG because, in the IDG model, there are three types of items involved and the interactions between the three types of items are different from that in \cite{LTLW2005_BGT}.

In the next section, we propose a probabilistic non-adaptive and a probabilistic two-stage adaptive pooling design and decoding algorithms for both the variants of the IDG model discussed  this section.
 
\section{Pooling designs and Decoding Algorithm} \label{sec3}
In this section, we propose a non-adaptive pooling design and decoding algorithm as well as a two-stage adaptive pooling design and decoding algorithm for the IDG-WSI Model. The pooling designs and decoding algorithms for the IDG-NSI model follows from those for the IDG-WSI Model by replacing $I_{max}$ by $r$.

{\it Non-adaptive pooling design:}
The pools are generated from the matrix $\mathbf{M}_{NA}\in \mathbb{B}^{T_{NA} \times n}$. The entries of $\mathbf{M}_{NA}$ are i.i.d. as ${\cal B}(p_1)$. Test the pools denoted by the rows of $\mathbf{M}_{NA}$. Let the outcome vector be given by $\mathbf{y} \in \mathbb{B}^{T_{NA} \times 1}$. The exact value of $T_{NA}$ is specified in (\ref{eqn-beta1}) and (\ref{eqn-beta_NA}) (where $T_{NA}=\beta_{NA} \log n$) in Sub-section \ref{sec3-subsec1}, and its scaling is given in Theorem \ref{thm-NA_UB} (which appears before the beginning of Sub-section \ref{sec3-subsec1}). The exact value of $p_1$ is also given in Theorem \ref{thm-NA_UB}.

{\it Adaptive pooling design:}
A set of pools are generated from the matrix $\mathbf{M}_{1} \in \mathbb{B}^{T_1 \times n}$ whose entries are i.i.d. as ${\cal B}(p_1)$. The pools denoted by the rows of $\mathbf{M}_1$ are tested first and all the defectives are classified from the outcome vector $\mathbf{y}_1 \in \mathbb{B}^{T_1 \times 1}$. Denote the number of items declared defectives by ${\hat{d}}$ and the set of declared defectives by $\left\{\hat{u}_1, \hat{u}_2, \cdots, \hat{u}_{\hat{d}}\right\}$. If $\hat{d} \neq d$, an error is declared. We keep these declared defectives aside and generate another pooling matrix $\mathbf{M}_2 \in \mathbb{B}^{T_2 \times (n-d)}$, whose entries are i.i.d. as ${\cal B}(p_2)$, for the rest of the items. Now, test the pools denoted by the rows of the matrix $\mathbf{M}_2$ along with each of the items declared defectives and the outcomes are denoted by $\mathbf{y}_{\hat{u}_1},\mathbf{y}_{\hat{u}_2},\cdots,\mathbf{y}_{\hat{u}_{{d}}} \in \mathbb{B}^{T_2 \times 1}$. The two stages of testing are done non-adaptively as represented in Fig. \ref{fig-adaptive_PD}, and hence the pooling scheme is a two-stage adaptive pooling design. The exact values of $p_1$ and $p_2$ are given in Theorem \ref{thm-A_UB} (which appears before the beginning of Sub-section \ref{sec3-subsec1}). The scaling of $T_1$ and $T_2$ are also given in Theorem \ref{thm-A_UB} and their exact values are given in (\ref{eqn-beta1}) and (\ref{eqn-beta2}) (where, $T_{i}=\beta_{i} \log n$). The total number of tests is given by $T_1+dT_2$.

\begin{figure}[htbp] 
\centering 
\includegraphics[totalheight=2.8in,width=3.6in]{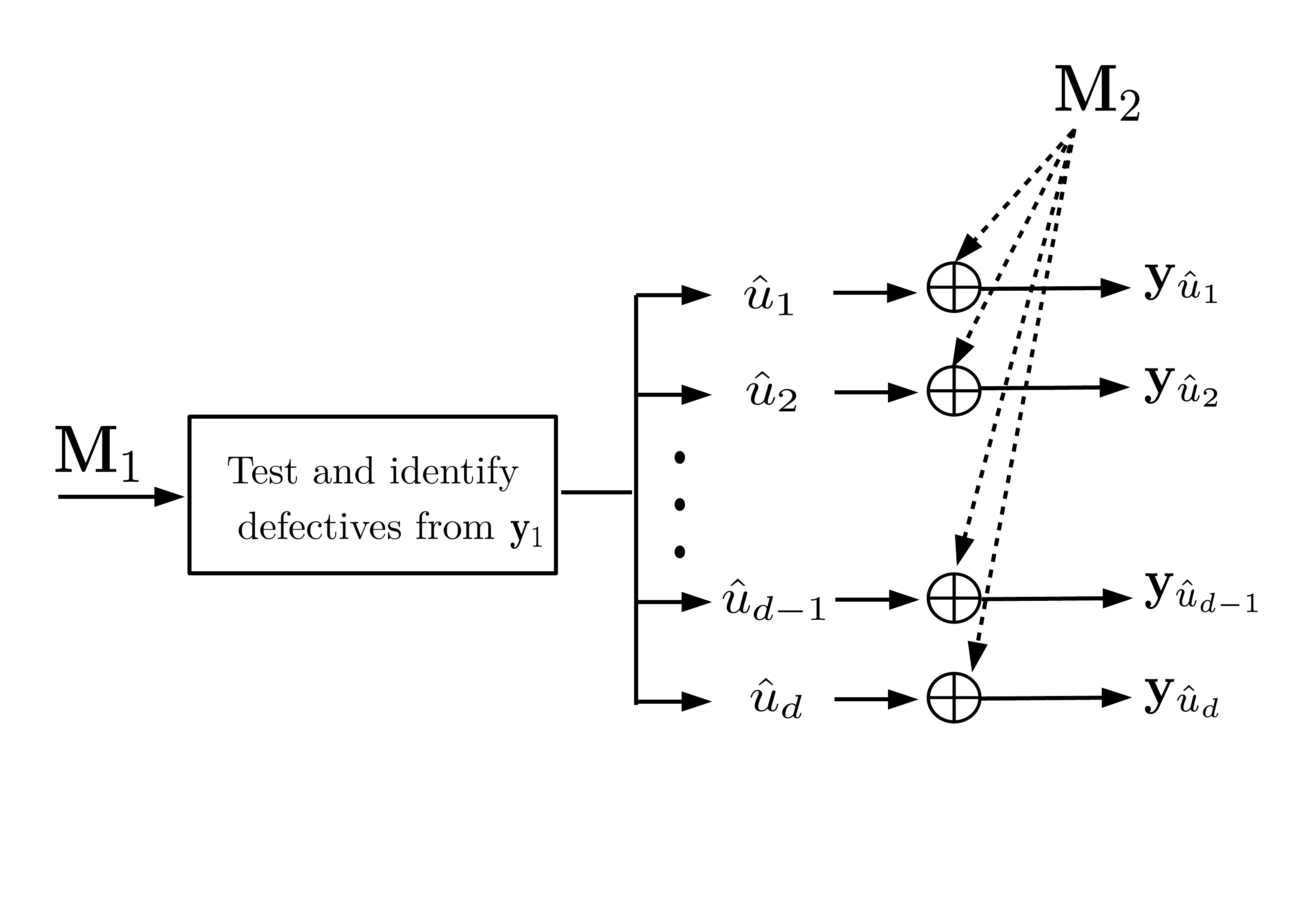}
\vspace{-1.5cm}
\caption{The proposed two-stage adaptive pooling design scheme is demonstrated here. The symbol $\bigoplus$ indicates that the pooling matrix $\mathbf{M}_2$ is tested along with the items $\hat{u}_i$ which are declared defectives. The items  non-associated with $\hat{u}_i$ are determined from the outcome vector $\mathbf{y}_{u_i}$, for $i=1,2,\cdots,d$.}
\label{fig-adaptive_PD}
\end{figure}

The defectives are expected to participate in a higher fraction of positive outcome tests than the normal items or the inhibitors. And, once the defectives are identified, tests of each one of them with rest of the items can be used to determine their associations. We show that this can be done non-adaptively as well. The decoding algorithm proceeds in two steps for both non-adaptive and adaptive pooling design. The first step will identify the defectives from the outcome vectors $\mathbf{y}$ and $\mathbf{y}_1$ in the non-adaptive and adaptive pooling designs respectively, according to the fraction of positive outcome tests in which an item participates. The second step will identify the inhibitors and their associations with the declared defectives using subsets of the outcome vector $\mathbf{y}$ in the non-adaptive pooling design and the outcome vectors $\mathbf{y}_{\hat{u}_1},\mathbf{y}_{\hat{u}_2},\cdots,\mathbf{y}_{\hat{u}_{d}}$ in the adaptive pooling design.

Let us define the following notations\footnote{From hereon, we reserve the notation $u$ to represent a defective, $v$ to represent a normal item and $s$ to represent an inhibitor.} with respect to the pools represented by $\mathbf{M}_{NA}$ and $\mathbf{M}_1$ which are eventually useful in characterizing the statistics of the different types of items that are used in the decoding algorithm.\\
{\it Notations:}
\begin{itemize}
\item ${\cal I}(u)$ denotes the set of inhibitors that the defective $u$ is associated with.
\item ${\mathscr F}_{u_{k}}$ denotes the event that none of the inhibitors associated with a defective $u_k$ appears in a test, given that the defective $u_k$ appears in the test.
\item ${\cal D}^{(j)}_i \subseteq {\cal P}(\{u_1,\cdots,u_d\})$ denotes the $j^{\text{th}}$-set in the (arbitrarily) ordered set of all $i$-tuple subsets of the defective set denoted by ${\cal D}_i$, for $j=1,\cdots,{d \choose i}$, where $u_i$ denotes a defective and ${\cal P}\{(u_1,\cdots,u_d)\}$ denotes the power set of the set of defectives.
\item ${\cal D}(s)$ denotes the defectives associated with the inhibitor $s$ and its complement is given by $\overline{{\cal D}(s)}={\cal D}-{\cal D}(s)$. 
\item  ${\overline{{\cal D}(s)}}_i$ denotes the (arbitrarily) ordered set of all $i$-tuple subsets of the defective set $\overline{{\cal D}(s)}$ and the $j^{\text{th}}$-set in $\overline{{\cal D}(s)}_i$ is denoted by ${\overline{{\cal D}(s)}}^{(j)}_i$.
\end{itemize}

\begin{example}\label{eg-realization-notation}
Realizations of the above notations for the association graph in Fig. \ref{fig-Eg1} considered in Example \ref{eg-bipartite} are given below. The inhibitor set is given by  ${\cal I}=\{s_1,\cdots,s_r\} \subset {\cal W}$ and the defective set is given by ${\cal D}=\{u_1,\cdots,u_d\} \subset {\cal W}$ with $r=d$. An inhibitor $s_i$ is associated with a distinct defective $u_i$, and so 
\begin{itemize}
\item ${\cal I}(u)$ for $u=u_i$ is given by ${\cal I}(u_i)=\{s_i\}$.
\item ${\mathscr F}_{u_{1}}$ represents the event that the inhibitor $s_1$ associated with the defective $u_1$ does not appear in a test, given that the defective $u_1$ appears in the test.
\item Realizations of ${\cal D}_i$ for $i=1,2$ are given by
\begin{align*}
{\cal D}_1=&\left\{\{u_1\},\{u_2\},\cdots,\{u_d\}\right\},\\
{\cal D}_2=&\left\{\{u_1,u_2\},\{u_1,u_3\},\cdots,\{u_1,u_d\},\right.\\ &~\left.\{u_2,u_3\},\cdots,\{u_2,u_d\},\cdots,\{u_{d-1},u_{d}\}\right\}.
\end{align*}Realizations of ${\cal D}^{(j)}_i$ for $(i,j)=(1,2)$ and $(i,j)=(2,3)$ are given by
\[
{\cal D}^{(2)}_1=\{u_2\},{\cal D}^{(3)}_2=\{u_1,u_4\}.
\]
\item ${\cal D}(s)$ for $s=s_1$ is given by ${\cal D}(s_1)=u_1$ and its complement is given by $\overline{{\cal D}(s_1)}=\{u_2,\cdots,u_d\}$.
\item Realizations of ${\overline{{\cal D}(s)}}_i$ for $s=s_1$ and $i=1,2$ 
\begin{align*}
{\overline{{\cal D}(s_1)}}_1=&\left\{\{u_2\},\cdots,\{u_d\}\right\},\\
{\overline{{\cal D}(s_1)}}_2=&\left\{\{u_2,u_3\},\{u_2,u_4\},\cdots,\{u_2,u_d\},\right.\\ &~\left.\{u_3,u_4\},\cdots,\{u_3,u_d\},\cdots,\{u_{d-1},u_{d}\}\right\}.
\end{align*}Realizations of ${\overline{{\cal D}(s)}}^{(j)}_i$ with $s=s_1$, for $(i,j)=(1,2)$ and $(i,j)=(2,3)$ are given by
\[
\overline{{\cal D}(s)}^{(2)}_1=\{u_3\},\overline{{\cal D}(s)}^{(3)}_2=\{u_2,u_5\}.
\]
\end{itemize}
\end{example}

We now define the following statistics corresponding to the different types of items. The following statistics also hold good when $\mathbf{y}_1$ is replaced by $\mathbf{y}$, as entries of both $\mathbf{M}_{NA}$ and $\mathbf{M}_1$ have the same statistics.
\begin{align}
\nonumber
&q^{(u)}_1 \triangleq \Pr \left\{\mathbf{y}_{1}(l)=1|\text{defective $u$ is present in the $l^{\text{th}}$-test}\right\}\\
\label{eqn-LB_q1}
&\geq(1-p_1)^{|{\cal I}(u)|} \geq (1-p_1)^{I_{max}}, \\
\nonumber
&q^{(v)}_2\triangleq \Pr \left\{\mathbf{y}_{1}(l)=1|\text{normal item $v$ is present in the $l^{\text{th}}$-test}\right\}\\
\label{eqn-q2}
&= \sum_{i=1}^{d} p^{i}_1(1-p_1)^{d-i}\sum_{j=1}^{d \choose i}\text{ Pr}\left\{ \underset{u_k \in {\cal D}^{(j)}_i}{\bigcup}{\mathscr F}_{u_{k}}\right\} \triangleq q_2\\
\label{eqn-q2_UB}
&\leq \sum_{i=1}^{d} p^{i}_1(1-p_1)^{d-i} {d \choose i}= 1-(1-p_1)^d \triangleq q^{UB}_{2},\\
\nonumber
&q^{(s)}_3 \triangleq \Pr \left\{\mathbf{y}_{1}(l)=1|\text{Inhibitor $s$ is present in the $l^{\text{th}}$-test}\right\}\\
\label{eqn-q3}
&= \sum_{i=1}^{\left|\overline{{\cal D}(s)}\right|} p^{i}_1(1-p_1)^{\left|\overline{{\cal D}(s)}\right|-i}\sum_{j=1}^{\left|\overline{{\cal D} (s)}_i\right|}\text{ Pr}\left\{  \underset{u_k \in {\overline{{\cal D}(s)}}^{(j)}_i}{\bigcup}{\mathscr F}_{u_{k}}\right\},\\\nonumber & \hspace{0.5cm} \text{ if $\left|\overline{{\cal D}(s)}\right| \geq 1$},\\
\nonumber
& =0, \text{ otherwise.}
\end{align} Since the outer and inner summations in (\ref{eqn-q3}) is over a subset of those in (\ref{eqn-q2}), $\underset{s}{\max} ~q^{(s)}_3 \leq q^{(v)}_2=q_2$. It is also intuitive that positive outcome for an inhibitor in a test is less probable than that for a normal item. The equality in (\ref{eqn-q2}) follows from the fact that a test outcome is positive iff at least one defective appears in the test (which is captured by the outer summation term) and none of the inhibitors associated with at least one of these defectives appears in the test (which is captured by the union of the events ${\mathscr F}_{u_k}$ over $u_k$). A similar explanation holds true for (\ref{eqn-q3}). The upper bound in (\ref{eqn-q2_UB}) follows from the upper bound of one on the probability terms of (\ref{eqn-q2}). In hindsight, the lower bound in (\ref{eqn-LB_q1}) and the upper bound in (\ref{eqn-q2_UB})  can be easily obtained as follows. The lower bound on the positive outcome statistics for a defective item in (\ref{eqn-LB_q1}) follows from the worst case statistics when all the inhibitors inhibit the expression of every defective. The upper bound on the statistics for a normal item in (\ref{eqn-q2_UB}) follows by using the best case positive outcome statistics, in the absence of inhibitors, where the appearance of any defective gives a positive test outcome. In the sequel, we shall exploit the difference between (\ref{eqn-LB_q1}) and (\ref{eqn-q2_UB}) to identify the defectives notwithstanding the fact the one of them could be loose bounds for specific association graphs. For example, (\ref{eqn-LB_q1}) is tight for the $1$-inhibitor model whereas (\ref{eqn-q2_UB}) could be a loose upper bound for the same association graph, depending on the values of $p$, $r$, and $d$. However, fortunately, $p_1$ can be chosen appropriately so that the looseness in the bounds do not affect the scaling of the upper bound on the number of tests required to identify the defectives, and the dominant scaling is determined by the number of tests required to identify the association pattern.

Denote the worst case negative outcome statistic for a defective by
\begin{align}\label{eqn-bmax}
 b_{max}=1-(1-p_1)^{I_{max}}.
\end{align}

Denote the set of tests corresponding to outcome vector $\mathbf{y}$ in which an item $w_j$ participates by ${\cal T}_{w_j}(\mathbf{y})$ and the set of positive outcome tests in which the item $w_j$ participates by ${\cal S}_{w_j}(\mathbf{y})$, for $j=1,2,\cdots,n$. The decoding algorithm is given as follows.

\begin{enumerate}
 \item {\em Step $1$ (Identifying the defectives for \underline{both non-adaptive} \underline{and adaptive pooling designs}):} \\For the non-adaptive pooling design, if $|{\cal S}_{w_j}(\mathbf{y})| > \left|{\cal T}_{w_j}(\mathbf{y})\right|[1-b_{max}(1+\tau))]$ with $b_{max}$ as defined in (\ref{eqn-bmax}), declare the item $w_j$ to be a defective. For the adaptive pooling design, we use the same criterion, replacing $\mathbf{y}$ by $\mathbf{y}_1$. Denote the number of items declared as defectives by ${\hat{d}}$ and the set of declared defectives by $\left\{\hat{u}_1, \hat{u}_2, \cdots, \hat{u}_{\hat{d}}\right\}$. If $\hat{d}\neq d$, declare an error. Denote the the remaining unclassified items in the population by $\left\{w'_1,\cdots,w'_{n-{d}}\right\} \triangleq \left\{w_1,\cdots,w_{n}\right\}-\left\{\hat{u}_1,\cdots, \hat{u}_{d}\right\}$.
 \item {\em Step $2$ (Identifying the inhibitors and their associations for \underline{non-adaptive pooling design}):} \\
 Let ${\cal P}_k$ denote the sets of pools in $\mathbf{M}_{NA}$ that contain only the declared defective $\hat{u}_k$ and none of the other declared defectives, for $k=1,\cdots,{d}$. Also, let the outcomes corresponding to these pools be positive. This  means that the pools in ${\cal P}_k$ do not contain any inhibitor from the set ${\cal I}(\hat{u}_k)$, which denotes the set of inhibitors associated with the item $\hat{u}_k$ if $\hat{u}_k$ is indeed a defective. Now, consider only the outcomes corresponding to these pools denoted by $\mathbf{y}_{{\cal P}_1} \subset \mathbf{y}, \cdots, \mathbf{y}_{{\cal P}_{d}} \subset \mathbf{y}$. The associations of the declared defectives are identified as follows. 
\begin{itemize}
 \item For each $k=1 \text{ to } {d}$, declare $(w'_j, \hat{u}_k)$ to be a non-associated inhibitor-defective pair if $w'_j$ participates in at least one of the tests corresponding to the outcome vector $\mathbf{y}_{{\cal P}_k}$ and declare the rest of the items to be associated with $\hat{u}_k$.
\end{itemize}The items declared as non-associated for all $k$ are declared to be be normal items. If ${\cal P}_k=\{\emptyset\}$ for some $k$, declare an error.
\item {\em Step $2$ (Identifying the inhibitors and their associations for \underline{adaptive pooling design}):} \\
Let ${\cal S}(\mathbf{y}_{\hat{u}_k})$ denote the set of positive outcome tests corresponding to $\mathbf{y}_{\hat{u}_k}$, i.e., these pools do not contain any inhibitor from the set ${\cal I}(\hat{u}_k)$ if $\hat{u}_k$ is a defective. 
\begin{itemize}
 \item For each $k=1 \text{ to } {d}$, declare $(w'_j, \hat{u}_k)$ to be a non-associated inhibitor-defective pair if $w'_j$ participates in at least one of the  tests in the set ${\cal S}(\mathbf{y}_{\hat{u}_k})$ and declare the rest of the items to be associated with $\hat{u}_k$.
\end{itemize}
The items declared as non-associated for all $k$ are declared to be be normal items. If ${\cal S}(\mathbf{y}_{\hat{u}_k})=\{\emptyset\}$ for some $k$, declare an error.
\end{enumerate}

The following toy example demonstrates the operation of the above decoding algorithm for non-adaptive pooling design.
\begin{example}\label{eg-NA_Algo}
Consider the following non-adaptive pooling design matrix $\mathbf{M}_{NA} \in \mathbb{B}^{5 \times 5}$ and the outcome vector $\mathbf{y}\in \mathbb{B}^{5 \times 1}$ for the underlying association graph shown in Fig. \ref{fig-NA_Algo}. The item $w_5$ is a normal item. Here, $r=d=2, n=5, T_{NA}=5$. 
 \begin{figure}[htbp]  
\centering
\includegraphics[totalheight=2in,width=3in]{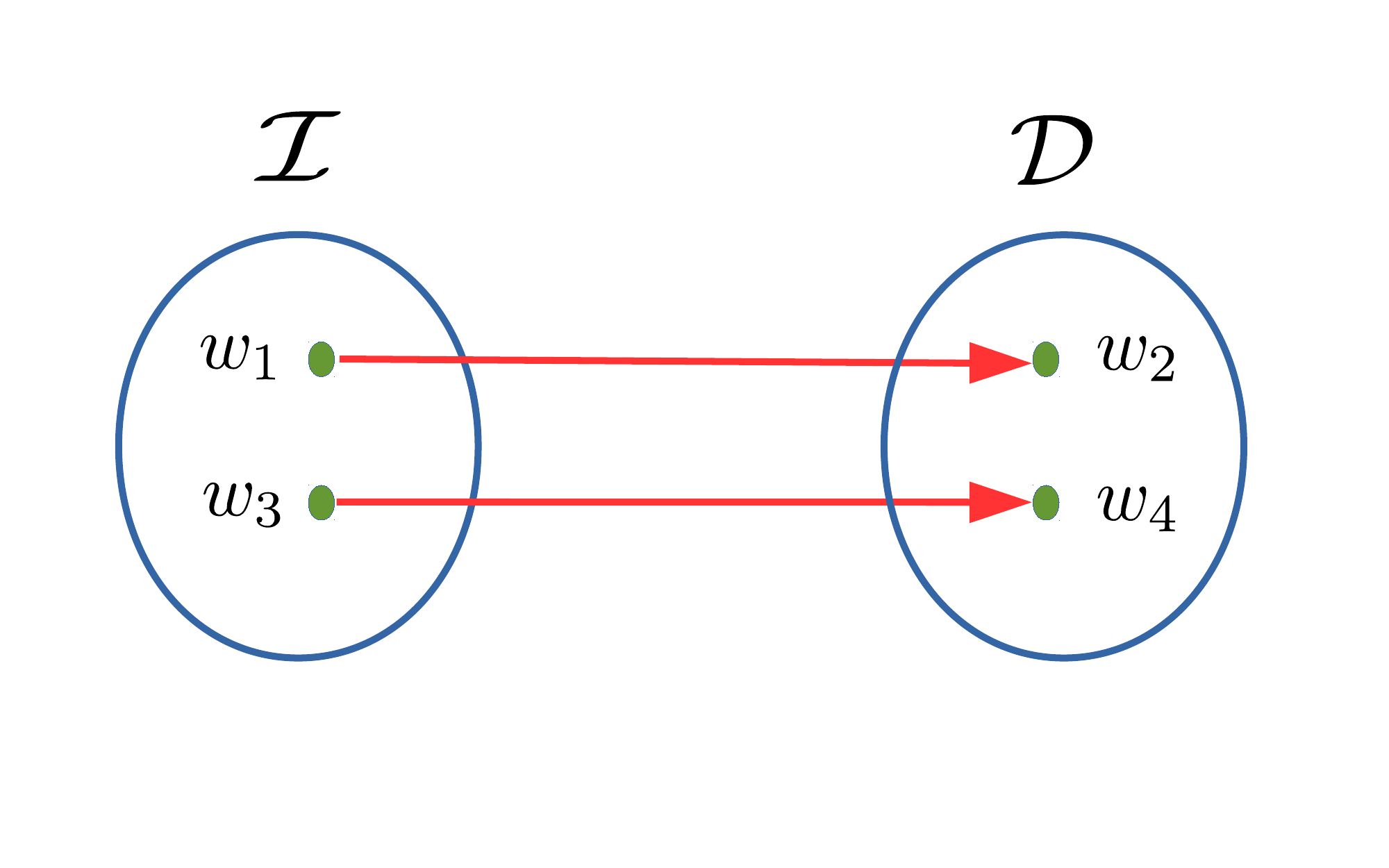}
\vspace{-1cm}
\caption{The underlying association graph for Example \ref{eg-NA_Algo}.}
\label{fig-NA_Algo}
\end{figure}
\begin{align*}
\mathbf{M}_{NA} = \begin{bmatrix}
                   1 & 1 & 0 & 0 & 1\\
                   0 & 1 & 0 & 1 & 0\\
                   0 & 1 & 1 & 0 & 1\\
                   1 & 0 & 0 & 1 & 1\\
                   0 & 0 & 1 & 1 & 1
                  \end{bmatrix}\Rightarrow \mathbf{y}=\begin{bmatrix}
					    0 \\ 1 \\ 1 \\ 1  \\0  
					   \end{bmatrix}
\end{align*}We recall that column-$j$ of the matrix $M_{NA}$ corresponds to the item $w_j$. The threshold for identifying the defectives in Step $1$ of the decoding algorithm is such that any item $w_j$ that satisfies the condition $\frac{|{\cal S}_{w_j}(\mathbf{y})|}{|{\cal T}_{w_j}(\mathbf{y})|}>\frac{1}{2}$ is declared to be a defective. Now, observe the operation of the decoding algorithm.

{\it Step $1$:} We observe that
\begin{align*}
&\frac{|{\cal S}_{w_1}(\mathbf{y})|}{|{\cal T}_{w_1}(\mathbf{y})|}=\frac{1}{2},\frac{|{\cal S}_{w_2}(\mathbf{y})|}{|{\cal T}_{w_2}(\mathbf{y})|}=\frac{2}{3},\frac{|{\cal S}_{w_3}(\mathbf{y})|}{|{\cal T}_{w_3}(\mathbf{y})|}=\frac{1}{2}, \\ & \frac{|{\cal S}_{w_4}(\mathbf{y})|}{|{\cal T}_{w_4}(\mathbf{y})|}=\frac{2}{3},\frac{|{\cal S}_{w_5}(\mathbf{y})|}{|{\cal T}_{w_5}(\mathbf{y})|}=\frac{1}{2}.
\end{align*}Items $w_2$ and $w_4$ are the only items that satisfy the condition $\frac{|{\cal S}_{w_j}(\mathbf{y})|}{|{\cal T}_{w_j}(\mathbf{y})|}>\frac{1}{2}$, and hence are declared defectives. Therefore, the declared defectives are given by $\hat{u}_1=w_2$, $\hat{u}_2=w_4$ and the remaining unclassified items are given by $w'_1=w_1, w'_2=w_3, w'_3=w_5$.

{\it Step $2$:} The ``useful'' pools used for identifying the ``non-associations'' are obtained as ${\cal P}_1=\{3\},{\cal P}_2=\{4\}$. This is because the third test outcome in which $\hat{u}_1$ participates and $\hat{u}_2$ does not participate is positive, and the fourth test outcome in which $\hat{u}_2$ participates and $\hat{u}_1$ does not participate is also positive. Since the items $w'_2$ and $w'_3$ participate in the third test, $(w'_2,\hat{u}_1)=(w_3,w_2)$ and $(w'_3,\hat{u}_1)=(w_5,w_2)$ are declared to be non-associated inhibitor-defective pairs and $(w'_1,\hat{u}_1)=(w_1,w_2)$ is declared to be an associated inhibitor-defective pair. Similarly, $(w'_1,\hat{u}_2)=(w_1,w_4)$ and $(w'_3,\hat{u}_1)=(w_5,w_4)$ are declared to be a non-associated item-pairs  and $(w'_2,\hat{u}_2)=(w_3,w_4)$ is declared to be an associated inhibitor-defective pair. Since the item $w'_3=w_5$ is declared to be non-associated with both $\hat{u}_1$ and $\hat{u}_2$, it is declared to be a normal item.

We emphasize that this is a toy example to demonstrate the operation of the proposed decoding algorithm and not representative of the values of $p$ or $\tau$ or $T_{NA}$ for the given values of $r,d,n$. 
\end{example}

\begin{remark}
 {\it (Step $1$)} The first step in the decoding algorithm, which is the same for both the non-adaptive and adaptive pooling design, is similar to the defective classification algorithm used in \cite{GEJS2014} for the $1$-inhibitor model. The underlying common principle used is that there exists statistical difference between the defective items and the rest of the items. Hence, with sufficient number of tests, the defectives can be classified by ``matching'' the tests in which an item participates and the positive outcome tests. The items involved in a large fraction of positive outcome tests are declared to be defectives. A similar decoding algorithm was used in the classical group testing framework with noisy tests \cite{CJSA_TIT2014}. Here, the inhibitors of a defective item, if any, behave like a noise due to probabilistic presence in a test. The (worst case) expected number of positive outcome tests in which a defective participates is at least $|{\cal T}_{w_j}(\mathbf{y})|[1-b_{max}]$. Like in \cite{GEJS2014}, the Chernoff-Hoeffding concentration inequality \cite{Hof1963} is used to bound the error probability and obtained the exact number of tests required to achieve a target (vanishing) error probability. {\it It is important to note that, a priori, it is not clear if a fixed threshold technique can sieve the defectives under worst case positive outcome statistics and the rest of the items under best case positive outcome statistics, with vanishing error probability}. The fact that this is indeed possible will be proved in the following sub-section.
\end{remark}

\begin{remark}
{\it (Step $2$)} In the IDG model, the inhibitors for each defective might be distinct. Hence, an inhibitor for one defective behaves as a normal item from the perspective of another defective. This defective-specific interaction is absent in the $1$-inhibitor model. So, any inhibitor can be identified using any defective, i.e, an inhibitor's behaviour is defective-invariant in the $1$-inhibitor model, which was exploited in identifying the inhibitors in \cite{GEJS2014}. Since each inhibitor's behaviour can be defective-specific in the IDG model, we need to identify the defectives first and then identify its associated inhibitors by observing the interaction of the other items with each of these defectives.
\end{remark}

The following theorems state the values of the parameters $p_1$, $p_2$, and $\tau$, and the scaling of the number of tests required for the proposed non-adaptive and adaptive pooling designs to determine the association graph with high probability. Similar results can be stated for the IDG-NSI model by replacing $I_{max}$ by $r$ in the following theorems.

\begin{theorem}[Non-adaptive pooling design]\label{thm-NA_UB}
 Choose the pooling design matrix $\mathbf{M}_{NA}$ of size $T_{NA} \times n$ with its entries chosen i.i.d. as ${\cal B}(p_1)$ with $p_1=\frac{1}{3(I_{max}+d)}$ for the IDG-WSI model. Test the pools denoted by the rows of the  matrix $\mathbf{M}_{NA}$ non-adaptively. The scaling of the number of tests sufficient to guarantee vanishing error probability (\ref{eqn-error_metric}) using the proposed decoding algorithm with  $\tau=\frac{1-b_{max}-q^{UB}_2}{2b_{max}}$ is given by $T_{NA}=O\left((I_{max}+d)^2\log n\right)$, where $q^{UB}_2$ and $b_{max}$  are defined in (\ref{eqn-q2_UB}) and (\ref{eqn-bmax}) respectively. 
\end{theorem}

\begin{theorem}[Adaptive pooling design]\label{thm-A_UB}
 Choose the pooling design matrices $\mathbf{M}_1$ and $\mathbf{M}_2$ of sizes $T_1 \times n$ and $T_2 \times n$ with its entries chosen i.i.d. as  i.i.d. ${\cal B}(p_1)$ and ${\cal B}(p_2)$ respectively, with $p_1=\frac{1}{3(I_{max}+d)}$ and $p_2=\frac{1}{2I_{max}}$ for the IDG-WSI model. Test the pools denoted by the rows of the matrices $\mathbf{M}_1$ non-adaptively and classify the defectives. Now, test each of the pools from $\mathbf{M}_2$ along with the $d$ classified defectives individually. The scaling of the number of tests sufficient to guarantee vanishing error probability (\ref{eqn-error_metric}) using the proposed decoding algorithm with  $\tau=\frac{1-b_{max}-q^{UB}_2}{2b_{max}}$ is given by $T_A=T_1+dT_2=O\left(I_{max}d\log n\right)$, where $q^{UB}_2$ and $b_{max}$  are defined in (\ref{eqn-q2_UB}) and (\ref{eqn-bmax}) respectively. 
\end{theorem}

\begin{remark}
 The value of $\tau=\frac{1-b_{max}-q^{UB}_2}{2b_{max}}$  chosen in the above theorems implies that the decoding algorithm declares item $w_j$ to be a defective if $\frac{|{\cal S}_{w_j}(\mathbf{y})|}{|{\cal T}_{w_j}(\mathbf{y})|},\frac{|{\cal S}_{w_j}(\mathbf{y_1})|}{|{\cal T}_{w_j}(\mathbf{y_1})|}>\frac{(1-b_{max})+q^{UB}_2}{2}$. This threshold is simply an average between the worse-case positive outcome statistic for a defective and
the best-case positive outcome statistic for a normal item or an inhibitor. The values of $p$ and $p_1$ are chosen so that the former is greater than the latter.
\end{remark}

The following sub-section constitutes the proof of the above theorems. The exact number of tests required to guarantee vanishing error probability for recovery of the association graph are also obtained. The proof is exactly the same for the IDG-NSI model, but replacing $I_{max}$ by $r$.

\subsection{Error Analysis of the Proposed Algorithm} \label{sec3-subsec1}
As mentioned in Section \ref{sec2}, we require that
\begin{align*}
\underset{{\cal I}, {\cal D},{\cal E}\left({\cal I}, {\cal D}\right)}{\max}  \Pr\left\{\left({\cal I}, {\cal D},{\cal E}({\cal I},{\cal D})\right)\neq\left(\hat{\cal I}, \hat{\cal D},\hat{\cal E}(\hat{\cal I}, \hat{\cal D}\right)\right\} \leq cn^{-\delta},
\end{align*}for some constant $c>0$ and fixed $\delta>0$. For the non-adaptive pooling design, we find the number of tests $T_{NA}$ required to upper bound the error probability of the first step of the decoding algorithm by $c_1n^{-\delta_1}$ and that of the second step of the decoding algorithm by $c_2n^{-\delta_2}$, for some constants $c_1$ and $c_2$. A similar approach is taken for the two-stage adaptive pooling design to find the number of tests $T_1$ and the value of $T_2$. Finally, the values of $\delta_1$ and $\delta_2$ are chosen so that the total error probability is upper bounded by $cn^{-\delta}$, for some constant $c$ and given $\delta>0$. 

\subsubsection{Error Analysis of the First Step}
Since the first step of the decoding algorithm is the same for both the non-adaptive and adaptive pooling design, the bounds on the number of tests obtained below for adaptive pooling design applies for the non-adaptive pooling design also. The three possible error events in the first step of the decoding algorithm for both non-adaptive and adaptive pooling design are given by
\begin{enumerate}
 \item A defective is not declared as one.
 \item A normal item is declared as a defective.
 \item An inhibitor is declared as a defective. 
\end{enumerate}

Clearly, the defective that has the largest probability of a negative outcome, given by $b_{1_{max}}=\underset{u}{\max}\left(1- ~q^{(u)}_1\right)$, has the largest probability of not being declared as a defective. So, with $T_1=\beta_1 \log n$, the probability of the first error event for all the defectives can be upper bounded (using the union bound over all defectives) as

{\small \vspace{-0.3cm}
\begin{align*}
&d \sum_{t=0}^{T_1} {T_1 \choose t} p^t_1 (1-p_1)^{T_1-t}\hspace{-0.5cm}\sum_{v=tb_{max}(1+\tau)}^{t} {t \choose v} b_{1_{max}}^v (1-b_{1_{max}})^{t-v}\\
&=d \sum_{t=0}^{T_1} {T_1 \choose t} p^t_1 (1-p_1)^{T_1-t}~~\times\\&\hspace{0.5cm}\sum_{v=tb_{1_{max}}+t(b_{max}-b_{1_{max}}+b_{max}\tau)}^{t} {t \choose v} b_{1_{max}}^v (1-b_{1_{max}})^{t-v}\\
&\overset{(a)}{\leq} d \sum_{t=0}^{T_1} {T_1 \choose t} p^t_1 e^{-2t{(b_{max}-b_{1_{max}}+b_{max}\tau)}^2} (1-p_1)^{T_1-t}\\
&\overset{(b)}{=} d \left[1-p_1+p_1 e^{-2{(b_{max}-b_{1_{max}}+b_{max}\tau)}^2}\right]^{\beta_1 \log n}\\
&\overset{(c)}{\leq} d\exp\left\{-\beta_1 p_1 \log n \left(1-e^{-2{(b_{max}-b_{1_{max}}+b_{max}\tau)}^2}\right) \right\} \leq n^{-\delta_1}\\
&\overset{(d)}{\Leftarrow}  d\exp\left\{-\beta_1 p_1 \log n \left(1-e^{-2}\right)(b_{max}-b_{1_{max}}+b_{max}\tau)^2\right\}\\ & ~~~~\leq n^{-\delta_1}\\
&\Rightarrow \beta_1 \geq \frac{\left(\frac{\ln d}{\ln n}+\delta_1\right)\ln 2}{p_1(1-e^{-2})(b_{max}-b_{1_{max}}+b_{max}\tau)^2},
\end{align*}}where $(a)$ follows from Chernoff-Hoeffding bound \cite{Hof1963}\footnote{If the term ${b_{max}(1+\tau)}>1$, then the probability of the error event under consideration is equal to zero. So, it can be assumed that ${b_{max}(1+\tau)}\leq 1$.}, $(b)$ follows from binomial expansion, $(c)$ follows from the fact that $1-c \leq e^{-c}$, and $(d)$ follows from the fact that $\left(1-e^{-2 x^2}\right) \geq \left(1-e^{-2}\right)x^2$, for $0<x<1$. Using the fact that $b_{1_{max}} \leq b_{max}$, where $b_{max}$ is defined in (\ref{eqn-bmax}), the following bound on $\beta_1$ suffices.
\begin{align}
 \label{eqn-beta1-def}
& \beta_1  \geq \frac{\left(\frac{\ln d}{\ln n}+\delta_1\right)\ln 2}{p_1(1-e^{-2})(b_{max}\tau)^2}.
\end{align}Similarly, to guarantee vanishing probability for the second error event (union-bounded over all normal items) and the third error event (union-bounded over all inhibitors), it suffices that 
\begin{align}
\nonumber
& \beta_1\geq\frac{\left(\frac{\ln (n-d-r)}{\ln n}+\delta_1\right)\ln 2}{p_1(1-e^{-2})\left(1-b_{max}(1+\tau)-q_2\right)^2}, \\
\label{eqn-beta1-inh}
& \beta_1\geq\frac{\left(\frac{\ln r}{\ln n}+\delta_1\right)\ln 2}{p_1(1-e^{-2})\left(1-b_{max}(1+\tau)-q^{(s)}_{3}\right)^2},
\end{align}Since $\underset{s}{\max}~q^{(s)}_{3}\leq q_2$ and $r=o(n)$, the bound in (\ref{eqn-beta1-inh}) is asymptotically redundant for all values of $\tau$. So, substituting the upper bound on $q_2$ defined in (\ref{eqn-q2_UB}) by $q^{UB}_2$, it suffices that
\begin{align}
\label{eqn-beta1-ndni}
& \beta_1\geq\frac{\left(\frac{\ln (n-d-r)}{\ln n}+\delta_1\right)\ln 2}{p_1(1-e^{-2})\left(1-b_{max}(1+\tau)-q^{UB}_2\right)^2}
\end{align}Now, the value of $\tau$  chosen to optimize the denominators of (\ref{eqn-beta1-def}) and (\ref{eqn-beta1-ndni}) is given by $\tau=\frac{1-b_{max}-q^{(UB)}_2}{2b_{max}}$. Therefore, we have
\begin{align}\label{eqn-beta-max_p}
 \beta_1 \geq & \max\left\{\frac{4\left(\frac{\ln d}{\ln n}+\delta_1\right)\ln 2}{p_1(1-e^{-2})\left(1-b_{max}-q^{UB}_2\right)^2},\right.\\ \nonumber &\hspace{1.5cm}\left. \frac{4\left(\frac{\ln (n-d-r)}{\ln n}+\delta_1\right)\ln 2}{p_1(1-e^{-2})\left(1-b_{max}-q^{UB}_2\right)^2} \right\}.
\end{align}The term $1-b_{max}-q_2$ can be lower bounded as follows.
\begin{align*}
 1-b_{max}-q_2 & \geq (1-p_1)^{I_{max}}-  (1-(1-p_1)^d)\\
 & \geq 1-(I_{max}+d)p_1.
\end{align*} The last lower bound above follows from the fact that \mbox{$(1-p_1)^{I_{max}} \geq (1-I_{max} p_1)$} and $(1-p_1)^d \geq (1-dp_1)$. Optimizing the denominator terms of (\ref{eqn-beta-max_p}) with respect to $p_1$, we have $p_1=\frac{1}{3(I_{max}+d)}$.  Hence, using $r,d=o(n)$ in (\ref{eqn-beta-max_p}), for sufficiently large $n$ it suffices that
\begin{align} \label{eqn-beta1}
 \beta_{NA},\beta_1 \geq \frac{27\left(I_{max}+d\right)\left(\frac{\ln (n-d-r)}{\ln n}+\delta_1\right)\ln 2}{(1-e^{-2})},
\end{align}where $T_{NA}=\beta_{NA} \log n$.

\subsubsection{Error Analysis of the Second Step}
In the error analysis of the second step, we assume that all the defectives have been correctly declared. Errors due to error propagation from the first step shall be analyzed later.\\
\underline{\it Non-adaptive pooling design:}\\
The only error event for the non-adaptive pooling design in the second step is that there  does not exist a set of pools ${\cal P}_k$ such that they contain only the defective $u_k$ and none of its associated inhibitors ${\cal I}(u_k)$, and all its non-associated items appear in at least one of such pools. Denote this error event by $\mathscr{U}(u_k)$. Clearly, none of the inhibitors associated with $u_k$ will be declared as non-associated with $u_k$. This follows from the definition of the set of pools ${\cal P}_k$ and the decoding algorithm.

The probability of the favourable event that a non-associated item appears along with a defective $u_k$, but none of its associated inhibitors and none of the other defectives appear in a pool from $\mathbf{M}_{NA}$ is given by \[b^{(u_k)}\triangleq p^2_1(1-p_1)^{|{\cal I}(u_k)|}(1-p_1)^{d-1}.\]Now, probability of the error event $\mathscr{U}(u_k)$ is upper bounded by 
\begin{align*}
 \Pr\{\mathscr{U}(u_k)\}&\leq(n-d-|{\cal I}(u_k)|)\left(1-b^{(u_k)}\right)^{T_{NA}}\\ &\leq (n-d-|{\cal I}(u_k)|) e^{-T_{NA}b^{(u_k)}} \leq n^{-\delta_2}, \text{ if}\\
& \hspace{-0.7in}~\beta_{NA} \geq \frac{\left(\frac{\ln (n-d-|{\cal I}(u_k)|)}{\ln n}+\delta_2\right)\ln 2}{b^{(u_k)}}.
\end{align*}Since $(1-p_1)^{|{\cal I}(u_k)|}\geq (1-p_1)^{I_{max}} \geq (1-I_{max}p_1)$ and $(1-p_1)^{d-1} \geq (1-dp_1)$, substituting for $p_1$, it suffices that
\begin{align}\label{eqn-beta_NA}
\beta_{NA} \geq \frac{81}{4} (I_{max}+d)^2 \left(\frac{\ln (n-d)}{\ln n}+\delta_2\right)\ln 2.
\end{align}\\
\underline{\it Adaptive pooling design:}\\
Like in non-adaptive pooling design the only error event, denoted by ${\mathscr E}{(u_k)}$, is that items $w_j$ not associated with $u_k$ are declared as associated inhibitors, i.e., the item $w_k$ does not appear in any of the positive outcome tests ${\cal S}(\mathbf{y}_{u_k})$.  Clearly, none of the inhibitors associated with $u_k$ will be declared as non-associated with $u_k$.

Let $T_2=\beta_2 \log n$. The number of tests required to guarantee vanishing error probability for the error event  ${\mathscr E}{(u_k)}$ is evaluated as follows. Let $w_j \notin {\cal I}(u_k)$. Define

\begin{align*}
a^{(u_k)}_{w_j} &\triangleq \text{ Pr}\left\{\mathbf{y}_{{u}_{k}}(l)=1\left|\right. w_j \text{ is present in $l^{\text{th}}$-test}\right\}\\
& \geq (1-p_2)^{|{\cal I}(u_k)|} \triangleq a^{(u_k)}.
\end{align*}Now, we have
\begin{align*} 
& \text{Pr}\left\{{\mathscr E}{(u_k)}\right\}\leq \left(n-d-|{\cal I}(u_k)|\right)\left(1-a^{(u_k)}p_2\right)^{T_2} \leq n^{-\delta_2}\\
 \Leftarrow ~& \beta_2 \geq  \frac{\left(\frac{\ln \left(n-d-|{\cal I}(u_k)|\right)}{\ln n}+\delta_2\right)\ln 2}{p_2 \left(1-p_2\right)^{|{\cal I}(u_k)|}}.
\end{align*}Using the fact that $\left(1-p_2\right)^{|{\cal I}(u_k)|} \geq 1-|{\cal I}(u_k)|p_2$, and substituting $p_2=\frac{1}{2I_{max}}$, we have the following bound.
\begin{align}\nonumber
&\beta_2 \geq  4 I_{max}\left(\frac{\ln \left(n-d-|{\cal I}(u_k)|\right)}{\ln n}+\delta_2\right)\ln 2\\
\label{eqn-beta2}
\Leftarrow ~&\beta_2 \geq  4 I_{max}\left(\frac{\ln \left(n-d\right)}{\ln n}+\delta_2\right)\ln 2.
\end{align}

\subsubsection{Analysis of Total Error Probability}

Assuming that the target total error probability is $O(n^{-\delta})$, the values of $\delta_1$ and $\delta_2$ need to be determined. Towards that end, define the following events.
\begin{align*}
{\mathscr E}_{ij} \triangleq & \text{ Event of declaring $(w_i,w_j), i\neq j$, to be an associated }\\&\text{ pair,}\\
{\mathscr W} \triangleq & \text{ Event that at least one actual defective has not been}\\&\text{ declared as a defective.}
\end{align*}Let ${\cal E}$ denote the correct association pattern for some realization $\{{\cal I},{\cal D}\}$. Now, the total probability of error is given by
\begin{align} \nonumber
&\text{Pr}\left\{\underset{(w_i,w_j)\notin {\cal E}}{\bigcup}{\mathscr E}_{ij}\bigcup {\mathscr W}\right\} \leq \sum_{(w_i,w_j)\notin {\cal E}}\text{Pr}\left\{{\mathscr E}_{ij}\right\}+\text{Pr}\left\{{\mathscr W}\right\}\\
\label{eqn-tpe1}
\leq &\sum_{w_i\neq w_j} \sum_{w_j\in {\cal N \cup \cal I}}\text{Pr}\left\{{\mathscr E}_{ij}\right\}+\sum_{w_i\notin {\cal I}(w_j)} \sum_{w_j\in {\cal D}}\text{Pr}\left\{{\mathscr E}_{ij}\right\}   \\ \nonumber &~~~~ +  \text{ Pr}\left\{{\mathscr W}\right\}\\
\label{eqn-tpe2}
< & n\times 2n^{-\delta_1}+ d n^{-\delta_2}  + n^{-\delta_1}.
\end{align}There are two possible ways in which the event ${\mathscr E}_{ij}$, for $(w_i,w_j)\notin {\cal E}$, can occur. One possibility is that the item $w_j$ has been erroneously declared as a defective in the first step of the algorithm, and hence any item $w_i$ declared to be associated with $w_j$ is an erroneous association. The first term in (\ref{eqn-tpe1}) represents this possibility. The other possibility is that $w_j$ has been correctly identified as a defective, but the item $w_i$ is erroneously declared to be associated with $w_j$. The second term in (\ref{eqn-tpe1}) represents this possibility. The last term accounts for the fact that a defective might be missed out in the first step of the algorithm. Note that the other two terms do not capture this error event. Finally, (\ref{eqn-tpe2}) follows from the error analysis of the first and second steps of the decoding algorithm. Therefore, if the target error probability is $O(n^{-\delta})$, then choose $\delta_1,\delta_2=\delta+1$.

Recall that the number of tests required for non-adaptive and adaptive pooling designs are given by $T_{NA}=\beta_{NA}\log n$ and $T_A=T_1+dT_2=(\beta_1 +d\beta_2) \log n$ respectively. Therefore, from (\ref{eqn-beta1}), (\ref{eqn-beta_NA}), and (\ref{eqn-beta2}) we have that $T_{NA}=O\left((I_{max}+d)^2 \log n\right)$ and $T_A=O\left(I_{max} d \log n\right)$.

\subsection{Adaptation for the IDG-NSI Model}

The only modification required in the pooling design and decoding algorithm proposed for the IDG-WSI model to adapt it for the IDG-NSI model is that $I_{max}$ is replaced by $r$. For the sake of clarity, we list the only changes below.

\begin{enumerate}
 \item The pooling design parameters are chosen as $p=p_1=\frac{1}{3(r+d)}$, $p_2=\frac{1}{2r}$.
 \item In Step $1$ of the decoding algorithm the threshold for identifying the defectives is chosen as $|{\cal S}_{w_j}(\mathbf{y}_1)| > |{\cal T}_{w_j}(\mathbf{y}_1)|[1-b_{max}(1+\tau))]$, where $b_{max}=1-(1-p_1)^r$. Intuitively, this worst-case threshold corresponds to a scenario where every inhibitor inhibits every defective, i.e., the $1$-inhibitor model.
 \item The values of $\beta_{NA}$, $\beta_1$ and $\beta_2$ are chosen as
 \begin{align*}
\beta_{NA}\geq &\max  \left\{\frac{27\left(r+d\right)\left(\frac{\ln (n-d-r)}{\ln n}+\delta_1\right)\ln 2}{(1-e^{-2})}\right.,\\
\nonumber
 & \left.\frac{81}{4} (r+d)^2 \left(\frac{\ln (n-d)}{\ln n}+ \delta_2\right)\ln 2  \right\},\\
\beta_1 \geq & \frac{27\left(r+d\right)\left(\frac{\ln (n-d-r)}{\ln n}+\delta_1\right)\ln 2}{(1-e^{-2})},\\
 \beta_2 \geq &  4 r\left(\frac{\ln \left(n-d\right)}{\ln n}+\delta_2\right)\ln 2.
 \end{align*}
\end{enumerate}
Hence, the total number of tests required for the IDG-NSI model scales as $T_{NA}=O\left((r+d)^2 \log n\right)$ for the non-adaptive pooling design and $T_A=O(rd \log n)$ for the two-stage adaptive pooling design.

In the next section, lower bounds on the number of tests for non-adaptive and adaptive pooling designs are obtained.

\section{Lower Bounds for Non-Adaptive and Adaptive Pooling Design} \label{sec4}
In this section, two lower bounds on the number of tests required for non-adaptive pooling designs for solving the IDG-NSI and IDG-WSI problems with vanishing error probability are obtained. One of the lower bounds is simply obtained by counting the entropy in the system and this lower bound also holds good for adaptive pooling designs. The other lower bound is obtained using a lower bound result for the $1$-inhibitor model which is stated below. We recall that all the inhibitors inhibit the expression of every defective in the $1$-inhibitor model.

\begin{theorem}[Th. $1$, \cite{GEJS2014}]\label{thm-LB_1_Inh_Model}
An asymptotic lower bound on the number of tests required for non-adaptive pooling designs in order to classify $r$ inhibitors amidst $d$ defectives and $n-d$ normal items in the $1$-inhibitor model is given by $\Omega\left(\frac{r^2}{d \log{\frac{r}{d}}}\log n\right)$, in the $d=o(r), r=o(n)$ regime\footnote{Though Theorem $1$ in \cite{GEJS2014} is stated for the classification of both the defectives and inhibitors in the $1$-inhibitor model, it is also valid for classification of inhibitors alone. This is because the entropy in the system is dominated by the number of inhibitors, in the large inhibitor regime.}.
\end{theorem}

The second lower bound in the following theorem dominates in the large inhibitor regime, i.e., the number of inhibitors is large compared to the number of defectives. It conveys the number of tests required to identify the inhibitors alone. Though the inhibitors outnumber the defectives, they can be identified only in the presence of an associated defective. So, the worst scenario (in terms of number of tests) happens when most inhibitors have to be identified using a single defective, or in other words, all of the inhibitors happen to inhibit a single defective. The third lower bound in the following theorem exploits the intuition gained from Step $2$ of the decoding algorithm for non-adaptive pooling design (given in Section \ref{sec3}). This lower bound is obtained by characterizing the minimum number of tests required to identify the associations of every defective. Since no two defectives might be associated with a single inhibitor, it is necessary that no two defectives participate in the same test from which the associations of a defective are identified. Otherwise, the non-associated defective masks the effect of association of the associated inhibitor-defective pair. This might result in wrongly declaring the associated inhibitor-defective pair to be non-associated.

Throughout this section, lowercase alphabets are used for defectives and inhibitors whose realizations are revealed by a genie and uppercase alphabets are used for those whose realizations are unknown.

\begin{theorem}\label{thm-LB_GTI_ID_NSI}
An asymptotic lower bound on the number of tests required for non-adaptive pooling designs for solving the IDG-NSI problem with vanishing error probability for $r,d=o(n)$ is given by 
\begin{align*}
\max\left\{\Omega\left((r+d)\log n+rd\right), \Omega\left(\frac{r^2}{\log r} \log n\right), \Omega(d^2)\right\}. 
\end{align*}
\end{theorem}
\begin{proof}
The proof for the first lower bound on the number of tests follows by lower bounding the total number of possible realizations of the sets of inhibitors, defectives, and association patterns.
\begin{align}
\nonumber
 T_{NA} &\geq  H\left({\cal I}, {\cal D},{\cal E}({\cal I},{\cal D})\right)\\
 \label{eqn-LB_ent}
 &=H({\cal D})+H\left({\cal I}|{\cal D})+H({\cal E}({\cal I},{\cal D})|({\cal I},{\cal D})\right)\\
\nonumber
 &=\log \left({n \choose d} {n-d \choose r} \sum_{i_1=1}^{d}\cdots \sum_{i_r=1}^{d} \prod_{j=1}^{r} {d\choose i_j}\right)\\
\nonumber
 &\geq \log \left({n \choose d} {n-d \choose r}  \prod_{k=1}^r {d \choose \frac{d}{2}}\right)\\
\label{eqn-LB_ent_2}
 &=\Omega \left((r+d) \log n + rd\right),
\end{align}where $i_j$ denotes the number of defectives that the $j^{\text{th}}$-inhibitor can be associated with, and the last step follows by using Stirling's lower bound ${d \choose \frac{d}{2}}\geq 2^\frac{d}{2}$. This lower bound on the number of tests is also valid for {\it adaptive pooling designs}.

The second lower bound for non-adaptive pooling designs is obtained as follows. Assume that it is required to identify the inhibitors alone. Clearly, this requires lesser number of tests than the problem of identifying the association graph. Since the objective is to satisfy the error metric in (\ref{eqn-error_metric}), the error probability criterion \begin{align} \label{eqn-error_metric_real}
\Pr\left\{\hat{\cal I}\neq {\cal I}\right\} \leq cn^{-\delta}
\end{align}has to be satisfied for all possible association patterns ${\cal E}$ on all possible realizations of $({\cal I},{\cal D})$. Let $PD\text{-}DA$ denote a pooling design, decoding algorithm tuple that satisfies (\ref{eqn-error_metric_real}), and $\mathscr P$ denotes the set of all such tuples. Further, let $T_{NA}(PD\text{-}DA, {\cal I},{\cal D}, {\cal E})$ denote the minimum number of tests required by $PD\text{-}DA$ to satisfy  (\ref{eqn-error_metric_real}) for a particular realization of $\left({\cal I},{\cal D}, {\cal E}\right)$. We now have to determine the lower bound $\underset{{\mathscr P}}{\inf} \underset{\left({\cal I},{\cal D}, {\cal E}\right)}{\sup} T_{NA}$. We now have
\begin{align*}
\inf_{\mathscr P} \sup_{\left({\cal I},{\cal D}, {\cal E}\right)} T_{NA} \geq \inf_{\mathscr P} \sup_{\left({\cal I},{\cal D}, {\cal E}'\right)} T_{NA},
\end{align*}
where ${\cal E}'$ denotes a specific class of association pattern represented in Fig. \ref{fig-ass_LB}.
\begin{figure}[htbp]  
\centering
\includegraphics[totalheight=2.5in,width=3.4in]{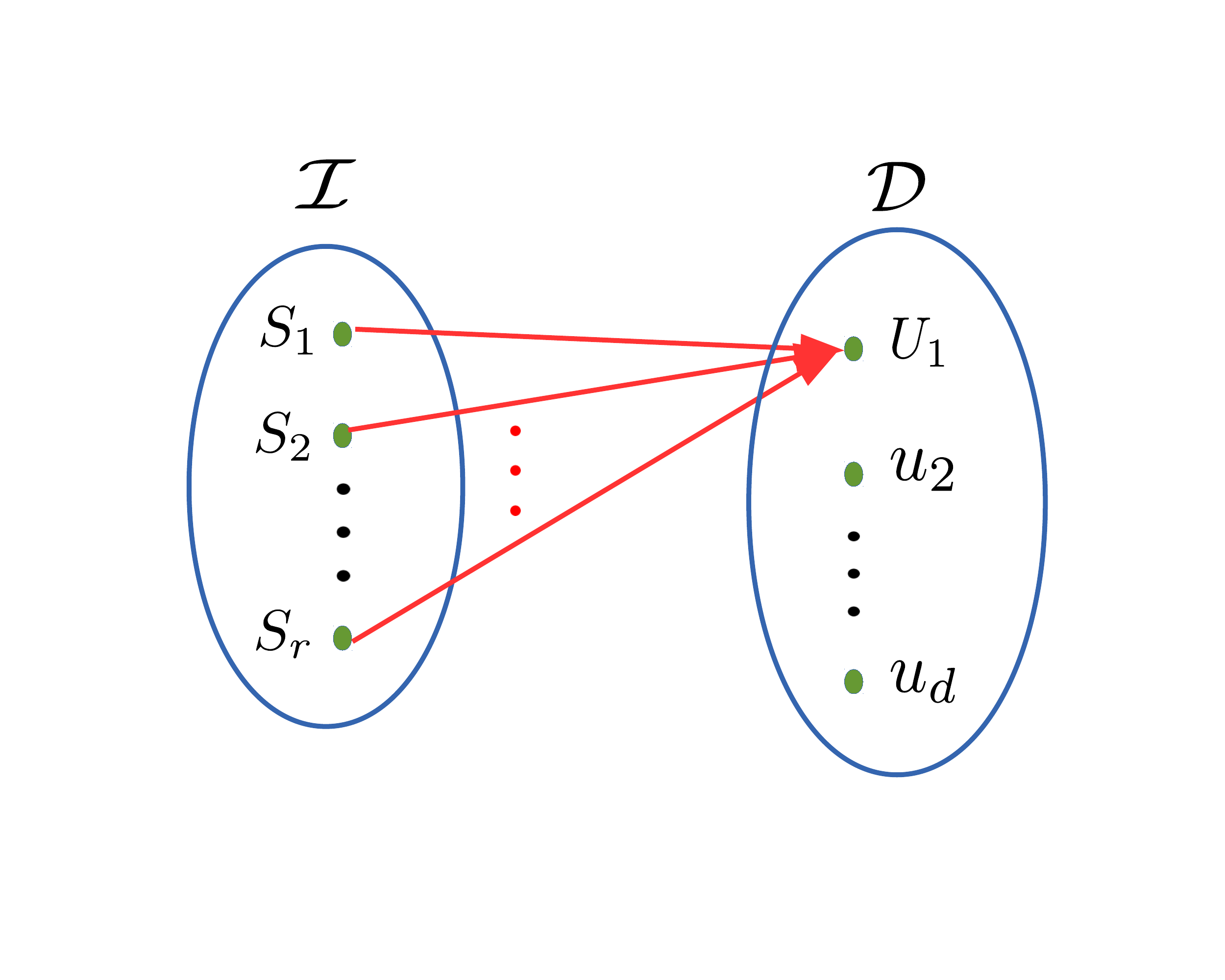}
\vspace{-1cm}
\caption{The class of association pattern used to obtain the second lower bound, illustrated for some realization of $({\cal I},{\cal D})$. A single defective is associated with all the inhibitors, but none of the other defectives are associated with any inhibitor.}
\label{fig-ass_LB}
\end{figure}Now, assume that a genie reveals the set of defectives ${\mathscr D}'\triangleq \{u_2,\cdots,u_d\}$ which are not associated with any of the inhibitors. A lower bound for this problem with side information from the genie is clearly a lower bound for the original problem. A lower bound on the number of tests $T'_{NA}$ for this problem is given by \cite{GEJS2014}\footnote{A similar expression is used to obtain Theorem $3$. This is derived formally using Fano's inequality. The steps involved are illustrated in the proof of the third lower bound.} 
\begin{align*}
 \sum_{l=1}^{T'_{NA}}H[\mathbf{y}(l)]=\Omega\left(\log {n-d \choose r}\right).
\end{align*}
Note that the presence of any defective from the set ${\mathscr D}'$ in a pool always gives a positive outcome, and hence provides zero information for distinguishing the inhibitors from the rest of the items as the entropy of such an outcome is zero. So, we assume that none of the tests contain items from  the set ${\mathscr D}'$. Therefore, the inhibitor identification problem for items with the association pattern as given in Fig. \ref{fig-ass_LB} is now reduced to the problem of identifying $r$ inhibitors amidst $n-d$ normal items and one defective item in the $1$-inhibitor model, where $d=o(n)$. For this problem, using Theorem \ref{thm-LB_1_Inh_Model}, it follows that the lower bound on the number of tests is given by $T'_{NA}= \Omega\left(\frac{r^2}{\log r} \log n\right)$. Hence, this is also a lower bound on the number of tests required to identify the association graph with vanishing worst case error probability.

The evaluation of the third lower bound is involved and is obtained as follows. Since the second lower bound is tighter when $r\geq d$, here we assume that $r<d$. Using similar arguments as in the second lower bound, a lower bound on the number of tests for the following reduced problem is a lower bound for the original problem. Let $\{S_2,\cdots,S_r\}$ denote a set of inhibitors associated with exactly one defective $U_d$. The defective $U_d$ is not inhibited by the inhibitor $S_1$. Further, the inhibitor $S_1$ is associated with exactly one of the defectives $\{U_1,\cdots,U_{d-1}\}$. This association graph is depicted in Fig. \ref{fig-LB_d^2}.
\begin{figure}[htbp]
\centering
\includegraphics[totalheight=2.8in,width=3.5in]{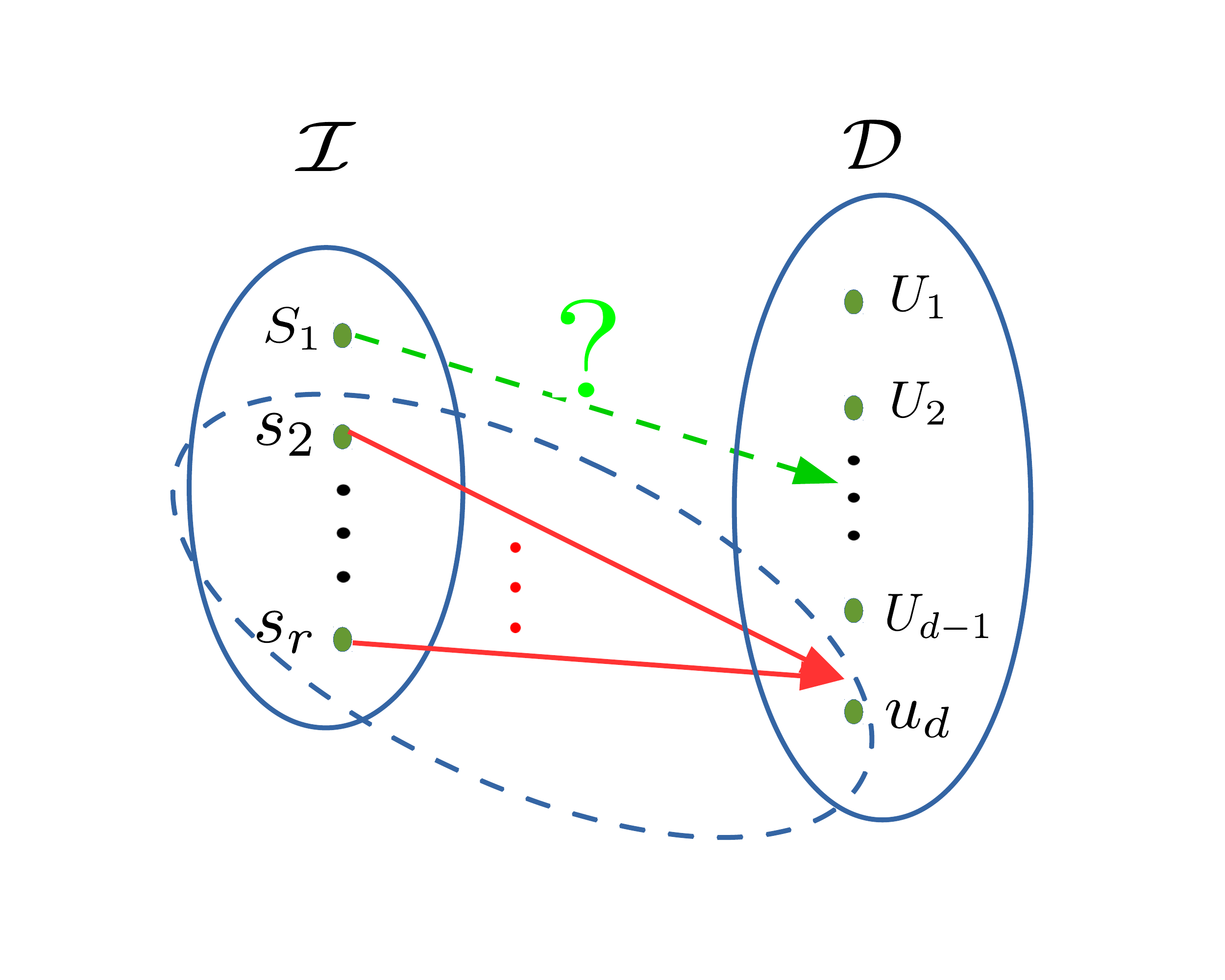}
\vspace{-0.5cm}
\caption{A possible association graph, where $r-1$ inhibitors and a single defective are associated only among themselves. The remaining inhibitor is associated with exactly one of the remaining defectives.}
 \label{fig-LB_d^2}
\end{figure}Let a genie reveal the set of inhibitors ${\mathscr I}_{-S_1}\triangleq\{s_2,\cdots,s_r\}$ and the defective $u_d$. The ``residual message'' in the system is now given by $W'\triangleq \left\{S_1,{\cal D}_{-u_d} ,{\cal E}\left(S_1,{\cal D}_{-u_d}\right)\right\}$, where ${\cal D}_{-u_d} \triangleq {\cal D}\backslash {u_d}$.

For the reduced problem, we have
\begin{align}\label{eqn-Pe_max}
&\underset{\underset{{\cal E}\left({S}_1, {{\cal D}}_{-u_d}\right)}{{S}_1, {{\cal D}}_{-u_d},}}{\max}  \Pr\left\{\left(\hat{S}_1, \hat{{\cal D}}_{-u_d},\hat{\cal E}\left(\hat{S}_1, \hat{{\cal D}}_{-u_d}\right)\right)\right.\\\nonumber& \left.\hspace{2cm}\neq \left({S}_1, {{\cal D}}_{-u_d},{\cal E}\left({S}_1, {{\cal D}}_{-u_d}\right)\right)\right\}\\
\label{eqn-Pe_avg}
\geq & ~\mathbb{E}_{f}\left[\Pr\left\{\left(\hat{S}_1, \hat{{\cal D}}_{-u_d},\hat{\cal E}\left(\hat{S}_1, \hat{{\cal D}}_{-u_d}\right)\right)\right.\right.\\\nonumber& \left.\left.\hspace{2cm}\neq \left({S}_1, {{\cal D}}_{-u_d},{\cal E}\left({S}_1, {{\cal D}}_{-u_d}\right)\right)\right\}\right]\triangleq P_{e_{avg}},
\end{align}where $f$ denotes a probability mass function of the association graph such that 
\begin{align}\label{eqn-f}
&\Pr\{(s_1,u)\in {\cal E}\left({s}_1, {{\mathscr D}}_{-u_d}\right)\}=\frac{1}{d-1}, u \in {\mathscr D}_{-u_d},\\
\nonumber
&\Pr\{{s}_1, {{\mathscr D}}_{-u_d}\}=\frac{1}{{n-r \choose d-1}(n-r-d+1)}
\end{align}for any realization of $(S_1,{\cal D}_{-u_d})$ given by $({s}_1, {{\mathscr D}}_{-u_d})$. So, a lower bound on the number of tests required to achieve vanishing average error probability $P_{e_{avg}}$ in (\ref{eqn-Pe_avg}) is also a lower bound on the number of tests required to achieve vanishing maximum error probability in (\ref{eqn-Pe_max}). These in turn give a lower bound on the number of tests for the original problem.

Using Fano's inequality, we have\footnote{For brevity, we omit the conditioning on ${\cal E}\left({\mathscr I}_{-S_1},u_d\right)$, which is also revealed by the genie, in the entropy and mutual information terms.}
{\begin{align}\nonumber
&H[{\cal E}\left(S_1,{\cal D}_{-u_d}\right)|S_1,{\cal D}_{-u_d},{\mathscr I}_{-S_1},u_d]\\ \label{eqn-sub1}&= \frac{1}{{n-r \choose d-1}(n-r-d+1)}\sum_{S_1,{\cal D}_{-u_d}}  \log (d-1) =\log (d-1) \\\nonumber
&\leq  1 + P_e H[{\cal E}\left(S_1,{\cal D}_{-u_d}\right)|S_1,{\cal D}_{-u_d},{\mathscr I}_{-S_1},u_d] \\\nonumber&+ I\left[{\cal E}\left(S_1,{\cal D}_{-u_d}\right);\mathbf{y}|S_1,{\cal D}_{-u_d},{\mathscr I}_{-S_1},u_d\right]\\
\label{eqn-cond_red_H}
&\leq  1 + P_e \log (d-1) + H\left[\mathbf{y}|S_1,{\cal D}_{-u_d},{\mathscr I}_{-S_1},u_d\right],
\end{align}}where (\ref{eqn-sub1}) is obtained using (\ref{eqn-f}), and $P_e\triangleq \Pr\{\hat{\cal E}\left(S_1,{\cal D}_{-u_d}\right)\neq {\cal E}\left(S_1,{\cal D}_{-u_d}\right)\}\leq P_{e_{avg}}$ denotes the average error probability in declaring the residual association pattern\footnote{The inequality holds because $\mathbb{E}\left[\mathbb{I}\left(\hat{\cal E}\left(S_1,{\cal D}_{-u_d}\right)\neq {\cal E}\left(S_1,{\cal D}_{-u_d}\right) \right)\right]$ \mbox{$\leq \mathbb{E}\left[\mathbb{I}\left(\left(\hat{S}_1,\hat{\cal D}_{-u_d},\hat{\cal E}\left(S_1,{\cal D}_{-u_d}\right)\right)\neq \left(S_1,{\cal D}_{-u_d},{\cal E}\left(S_1,{\cal D}_{-u_d}\right)\right) \right)\right]$}, where $\mathbb{E}[.]$ and $\mathbb{I}(.)$ represent the expectation operator and the indicator function respectively.}. The summation term in (\ref{eqn-sub1}) denotes summation over all possible realizations of $S_1,{\cal D}_{-u_d}$. Using the fact that conditioning reduces entropy in (\ref{eqn-cond_red_H}), we have
\begin{align} \nonumber
&{\footnotesize \sum_{l=1}^{T_{NA}}} H\left[\mathbf{y}(l)|S_1,{\cal D}_{-u_d},{\mathscr I}_{-S_1},u_d\right]\\
\label{eqn-lb_no_tests}
&\geq (1- P_e)\log (d-1)  -1.
\end{align}
The presence of items from the set $\{{\mathscr I}_{-S_1},u_d\}$ in a test can either reduce the entropy or leave the entropy of the test outcome unaffected. So, we consider only pooling designs that do not contain any item from the set $\{{\mathscr I}_{-S_1},u_d\}$. Therefore, the entropy of a test is dependent only on the realization of $S_1,{\cal D} \backslash {u_d}$, i.e.,
\begin{align*}
&  H\left[\mathbf{y}(l)|S_1,{\cal D}_{-u_d},{\mathscr I}_{-S_1},u_d\right]\\
=&H\left[\mathbf{y}(l)|S_1,{\cal D}_{-u_d}\right]\\=&\sum_{S_1,{\mathscr D}_{-u_d}}\Pr\left\{s_1,{\mathscr D}_{-u_d}\right\} H\left[\mathbf{y}(l)|s_1,{\mathscr D}_{-u_d}\right].
\end{align*}

Suppose that we are given a pool of $g_l$ items for the $l^\text{th}$ test. The entropy of the $l^\text{th}$ test outcome is non-zero only for those realizations of $S_1,{\cal D} \backslash {u_d}$ for which the $l^\text{th}$ pool contains exactly one defective and the inhibitor. This is because, otherwise, there is no randomness in the test outcome. There are $g_l(g_l-1){n-r-g_l \choose d-2}$ such possible realizations for $2\leq g_l\leq (n-r-d+2)$, and none for $g_l=0,1$ and for $g_l> (n-r-d+2)$. For each of these realizations, with $2\leq g_l\leq n-r-d+2$, the entropy of the test outcome is given by \[h\triangleq \frac{1}{d-1}\log (d-1)+\frac{d-2}{d-1}\log \left(\frac{d-1}{d-2}\right).\] Therefore, we have
\begin{align}\nonumber
&H\left[\mathbf{y}(l)|S_1,{\cal D}_{-u_d}\right]\\\nonumber=&\frac{1}{{n-r \choose d-1}(n-r-d+1)} g_l(g_l-1){n-r-g_l \choose d-2} h\\ \label{eqn-LB_sub_gopt}
<&  \frac{1}{{n-r \choose d-1}(n-r-d+1)} g^2_l {n-r-g_l \choose d-2} h.
\end{align}The term dependent on $g_l$ is re-written as \[g^2_l {n-r-g_l \choose d-2}=\frac{1}{(d-2)!}f(g_l),\] where \[f(g_l)\triangleq g^2_l \prod_{j=0}^{d-3} (n-r-g_l-j). \]We now maximize the above term with respect to $g_l\in [2,n-r-d+2]$ to obtain a lower bound on the number of tests. The following lemma gives the approximate optimum value of $g_l$ (denoted by $g_{opt}$). It is shown that $f(g_l) > f(g_l+\epsilon)$ for all $g_l \geq g_{opt}$ and $0<\epsilon\leq 1$, and  $f(g_l-\epsilon) < f(g_l)$ for all $g_l \leq g_{opt}$. Since $g_{opt}$ is independent of $l$, hereupon the subscript $l$ is dropped. It must be noted that $(g+\epsilon),(g-\epsilon)\in [2,n-r-d+2]$.
\begin{lemma}
There exists $n_0$ so that for all $n \geq n_0$, the optimum value of $g$ that maximizes $f(g)$ is given by $g_{opt}\triangleq \frac{2n}{d}+k$, where $k=o\left(\frac{2n}{d}\right)$.
\end{lemma}
\begin{proof}
To ensure $f(g) > f(g+\epsilon)$ for all $g\geq \frac{2n}{d-2}\triangleq g_1$ and $0<\epsilon \leq 1$, it suffices that
\begin{align}\nonumber
&g^2 \prod_{j=0}^{d-3} (n-r-g-j) > (g+\epsilon)^2 \prod_{j=0}^{d-3} (n-r-g-\epsilon-j)\\\nonumber
\Leftrightarrow &  \prod_{j=0}^{d-3} \left(1+ \frac{\epsilon}{n-r-g-\epsilon-j}\right) > \left(1+\frac{\epsilon}{g}\right)^2\\\nonumber
\Leftarrow &   \left(1+ \frac{\epsilon}{n-r-g-\epsilon}\right)^{d-2} > \left(1+\frac{\epsilon}{g}\right)^2\\\nonumber
\Leftarrow &   \left(1+ \frac{\epsilon}{n-r-g-\epsilon}\right)^{d-2} > e^{\frac{2\epsilon}{g}}\\
\label{eqn-ineq_gopt}
\Leftrightarrow &   \frac{g(d-2)}{2\epsilon}\ln \left(1+ \frac{\epsilon}{n-r-g-\epsilon}\right)> 1
\end{align}Since the above function is increasing in $g$, it is sufficient to prove that the above inequality is satisfied for $g=g_1$. Note that, since $r=o(n)$ and $d \underset{n\rightarrow \infty}{\longrightarrow} \infty$, we have $n-r-g_{1}=\Omega (n)$. So, in order to satisfy (\ref{eqn-ineq_gopt}) for all $n \geq n_0$ and some finite positive integer $n_0$ at $g=g_1$,  it suffices that
\begin{align*}
\frac{1}{1-\frac{2}{d-2}-\frac{2}{n}} > 1,
\end{align*}which is true. The above inequality follows by using the approximation $\ln (1+x)\sim x$ in (\ref{eqn-ineq_gopt}), for $x<<1$.

To ensure $f(g-\epsilon) < f(g)$ for all $g\leq \frac{2n}{d} -4 \triangleq g_{2}$, it is required that
\begin{align}
&(g-\epsilon)^2 \prod_{j=0}^{d-3} (n-r-g+\epsilon-j) < g^2 \prod_{j=0}^{d-3} (n-r-g-j)\\\nonumber
\Leftrightarrow &  \prod_{j=0}^{d-3} \left(1+ \frac{\epsilon}{n-r-g-j}\right) < \left(1+\frac{\epsilon}{g-\epsilon}\right)^2\\\nonumber
\Leftarrow &   \left(1+ \frac{\epsilon}{n-r-g-d+3}\right)^{d-2} < \left(1+\frac{\epsilon}{g-\epsilon}\right)^2\\\nonumber
\Leftarrow &  \exp \left( \frac{(d-2)\epsilon}{n-r-g-d+3}\right) < \left(1+\frac{\epsilon}{g-\epsilon}\right)^2\\
\label{eqn-ineq_gopt1}
\Leftrightarrow &   \frac{2(n-r-g-d+3)\ln \left(1+\frac{\epsilon}{g-\epsilon}\right)}{(d-2)\epsilon}> 1.
\end{align}Since the above function is decreasing in $g$, it is sufficient to prove that the above inequality is satisfied for $g=g_2$. In order to satisfy (\ref{eqn-ineq_gopt1}) for all $n \geq n_0$ and some finite positive integer $n_0$ at $g=g_2$,  it suffices that
\begin{align} \label{eqn-ineq_gopt2}
& \frac{2(n-r-g_2-d+3)\epsilon}{\left(g_{2}-{\epsilon}\right)(d-2)\epsilon} >1\\ \nonumber
\Leftrightarrow & \frac{(1-\frac{r}{n}-\frac{2}{d}-\frac{d}{n}+\frac{7}{n})}{\left(1-\frac{d(4+\epsilon)}{2n}\right)\left(1-\frac{2}{d}\right)} >1\\\nonumber
\Leftrightarrow & \frac{(1-\frac{r}{n}-\frac{2}{d}-\frac{d}{n}+\frac{7}{n})}{1-\frac{2}{d}-\frac{2d+0.5\epsilon}{n}+\frac{4+\epsilon}{n}} >1,
\end{align}which is true because $r <d$\footnote{Recall that this was assumed at the beginning of the proof of the third lower bound.}. The inequality (\ref{eqn-ineq_gopt2}) is obtained by using the approximation $\ln (1+x)\sim x$ in (\ref{eqn-ineq_gopt1}), for $x<<1$.

Therefore, we have $g_{opt} \in \left[\frac{2n}{d}-4,\frac{2n}{d-2}\right]$, and so $g_{opt}=\frac{2n}{d}+k$, for $k=o\left(\frac{2n}{d}\right)$.
\end{proof}

From (\ref{eqn-LB_sub_gopt}) and (\ref{eqn-lb_no_tests}), using the approximation $h \approx \frac{1}{d} \log d$, an asymptotic lower bound on the number of tests for vanishing error probability is given by
\begin{align}\label{eqn-scaling}
T_{NA} \geq \Omega\left(d\frac{{n-r \choose d-1}(n-r-d+1)}{ g^2_{opt} {n-r-g_{opt} \choose d-2}}\right).
\end{align}We now show that the fractional term above scales as $d$.

\begin{align}\nonumber
&\frac{{n-r \choose d-1}(n-r-d+1)}{ g^2_{opt} {n-r-g_{opt} \choose d-2}}\approx \frac{{n-r \choose d-1}(n-r-d+1)}{ \frac{4n^2}{d^2} {n-r-g_{opt} \choose d-2}}\\
\nonumber
= &\frac{{\prod_{i=0}^{d-2}(n-r-i)}(n-r-d+1)}{ \frac{4n^2}{d^2} (d-1) {\prod_{i=0}^{d-3}(n-r-g_{opt}-i)}}\\
\label{eqn-LB_approx1}
\approx  &\frac{d}{4} \prod_{i=0}^{d-3}\frac{n-r-i}{n-r-g_{opt}-i}\\
\nonumber
=& \frac{d}{4} \prod_{i=0}^{d-3}\left(1+\frac{g_{opt}}{ n-r-g_{opt}-i}\right) \\
\label{eqn-LB_approx2}
\geq & \frac{d}{4} \left(1+\frac{g_{opt}}{ n-r-g_{opt}}\right)^{d-2} \approx  \frac{d}{4}e^{\frac{g_{opt}(d-2)}{ n-r-g_{opt}}}  \approx  d\frac{e^2}{4}.
\end{align}
It must be noted that the ratio notion of approximation does not affect the scaling of the number of tests. The approximations in (\ref{eqn-LB_approx1}) and (\ref{eqn-LB_approx2}) make use of the fact that $r,d=o(n)$ and $g_{opt}=\frac{2n}{d}+o\left(\frac{2n}{d}\right)$.  Therefore, from (\ref{eqn-scaling}) and (\ref{eqn-LB_approx2}), we have $T_{NA}=\Omega(d^2)$.
\end{proof}

The lower bounds for the IDG-WSI model are obtained in the following theorem. Since we are interested in asymptotic lower bounds, we assume that the limits $\lt \frac{I_{max}}{r}$ and $\lt \frac{r}{dI_{max}}$ exist, and $I_{max} \underset{n \rightarrow \infty}{\longrightarrow} \infty$. The ideas used to obtain the following theorem are similar to those used in Theorem \ref{thm-LB_GTI_ID_NSI}. However, the ``second constraint'' (mentioned in the proof of the following theorem) needs to be accounted for.

\begin{theorem}\label{thm-LB_GTI_ID_WSI}
An asymptotic lower bound on the number of tests required for non-adaptive pooling designs for solving the IDG-WSI problem with vanishing error probability for $r,d=o(n)$ is given by 
\begin{align*}
\max\left\{\Omega\left((r+d)\log n+I_{max}d\right), \Omega\left(\frac{I^2_{max}}{\log I_{max}} \log n\right)\right\}. 
\end{align*} An additional asymptotic lower bound is given by $\Omega(d^2)$ when either $1)$ $r= (c-1)d + kd$, for some constant $0<k<1$ and $I_{max}=c$ or $2)$ $r= (c-1)d + k$ and $(c-1)d< r \leq cd$, for positive integer $k=o(d)$ and $I_{max}=c$ or $3)$ $(c-1)d< r \leq cd$ and $I_{max} \geq c+1$.
\end{theorem}
\begin{proof}
The first lower bound is obtained by lower bounding $H\left({\cal E}({\cal I},{\cal D})|({\cal I},{\cal D})\right)$ in (\ref{eqn-LB_ent}) as follows. Two constraints need to be satisfied while counting the entropy of association pattern. 
\begin{itemize}
 \item {\it First constraint:} Minimum degree of a vertex in ${\cal I}$ is one.
 \item {\it Second constraint:} Maximum degree of a vertex in ${\cal D}$ is no more than $I_{max}$.
\end{itemize}
We now consider the three possible cases below and show that in each of the cases the lower bound on the number of association patterns scales exponentially in $I_{max}d$. Let \mbox{$(c-1)d<r\leq cd$}, for some positive integer $c$, and so $I_{max}\geq c$. Define $\alpha_1=\lt \frac{c}{I_{max}}$ and $\alpha_2=\lt \frac{I_{max}}{r}$.

{\it Case $1$:} $\alpha_1 <1$ and $\alpha_2 < 1$. There exist positive constants $\beta_1<1$ and $\beta_2<1$ so that $c \leq \beta_1 I_{max}$ and $I_{max} \leq \beta_2 r$, $\forall n \geq n_0$. Define an association pattern, where each defective starting from $u_1$ is assigned a disjoint set of $c$ inhibitors until every inhibitor is covered. Therefore we have, $\underset{u_i \in {\cal D}}\max ~|{\cal I}(u_i)| \leq c$. Since the first constraint is satisfied, each defective is now free to choose an association pattern so that $\underset{u_i \in {\cal D}}\max ~|{\cal I}(u_i)| \leq I_{max}$. The number of such possible association patterns can be lower bounded by
\begin{align*}
 {r \choose I_{max}-c}^d \geq \left(\frac{r}{I_{max}-c}\right)^{(I_{max}-c)d}.
\end{align*}Thus, the entropy of association pattern  in this case scales (asymptotically) as the logarithm of the above quantity, which is given by $\Omega(I_{max}d)$.

{\it Case $2$:} $\alpha_1 <1$ and $\alpha_2 = 1$.  There exist positive constants $\beta_1<1$ and $\beta_2\leq 1$ with $\beta_2>\beta_1$ so that $c \leq \beta_1 I_{max}$ and $I_{max} \geq \beta_2 r$, $\forall n \geq n_0$. So, we have $I_{max}-c \geq (\beta_2 -\beta_1) r$, $\forall n \geq n_0$. Using similar arguments as in Case $1$, where after satisfying the first constraint, $\beta_2 r - c$ inhibitors are chosen to associate with each defective, we now have that the entropy of association pattern in this case scales asymptotically as $\Omega(rd)=\Omega({I_{max}d})$.

{\it Case $3$:}  $\alpha_1 =1$. Note that this case constitutes a large inhibitor regime with respect to the number of defectives (because $I_{max}\rightarrow \infty$). There exists a positive constant $\beta_1 \leq 1$ so that $c \geq \beta_1 I_{max}$, $\forall n \geq n_0$. The number of ways of assigning each defective to a disjoint set of $(c-1)$ inhibitors is given by

{\small\begin{align*}
 &{r \choose c-1}{r-(c-1) \choose (c-1)}{r-2(c-1) \choose (c-1)}\cdots {r-(d-1)(c-1) \choose (c-1)}\\
 &=\frac{r!}{((c-1)!)^d (r-d(c-1))!}\\
 &\underset{(a)}{\geq} \frac{\sqrt{2\pi}r^{r+\frac{1}{2}}e^{-r}}{e^2 (c-1)^{d(c-1+\frac{1}{2})}e^{-d(c-1)}(r-d(c-1))^{r-d(c-1)+\frac{1}{2}} e^{-(r-d(c-1))}}\\
 &\underset{(b)}{\geq} \frac{\sqrt{2\pi}(d(c-1))^{d(c-1)+\frac{1}{2}}e^{-dc}}{e^2 (c-1)^{d(c-1+\frac{1}{2})}e^{-d(c-1)}(r-d(c-1))^{r-d(c-1)+\frac{1}{2}} e^{-(r-d(c-1))}}\\
 &\underset{(c)}{\geq} \frac{\sqrt{2\pi}(d(c-1))^{d(c-1)+\frac{1}{2}}e^{-dc}}{e^2 (c-1)^{d(c-1+\frac{1}{2})}e^{-d(c-1)}d^{d+\frac{1}{2}}}\\
 &=\frac{\sqrt{2\pi}(c-1)^{\frac{1}{2}}d^{d(c-2)}}{e^2(c-1)^{\frac{d}{2}}e^d}=\frac{\sqrt{2\pi}(c-1)^{\frac{1}{2}}d^{d(c-2)}}{e^2 d{\frac{d\log_d(c-1)}{2}}d^{d\log_d e}}\\
 &=\frac{\sqrt{2\pi}}{e^2}(c-1)^{\frac{1}{2}}d^{d\left(c-\frac{\log_d(c-1)}{2}-\log_d e-2\right)},
\end{align*}}where $(a)$ follows from Stirling's lower and upper bounds for factorial functions, $(b)$ and $(c)$ follow from the fact that $d(c-1) < r \leq cd$. Observe that the remaining $r-d(c-1)$ inhibitors can be assigned one each to one defective without violating the second constraint. Thus, the entropy of association pattern in this case scales asymptotically as $\Omega(cd)=\Omega(\beta_1 I_{max}d)=\Omega({I_{max}d})$.


\begin{figure}[htbp]  
\centering
\includegraphics[totalheight=2.5in,width=3.4in]{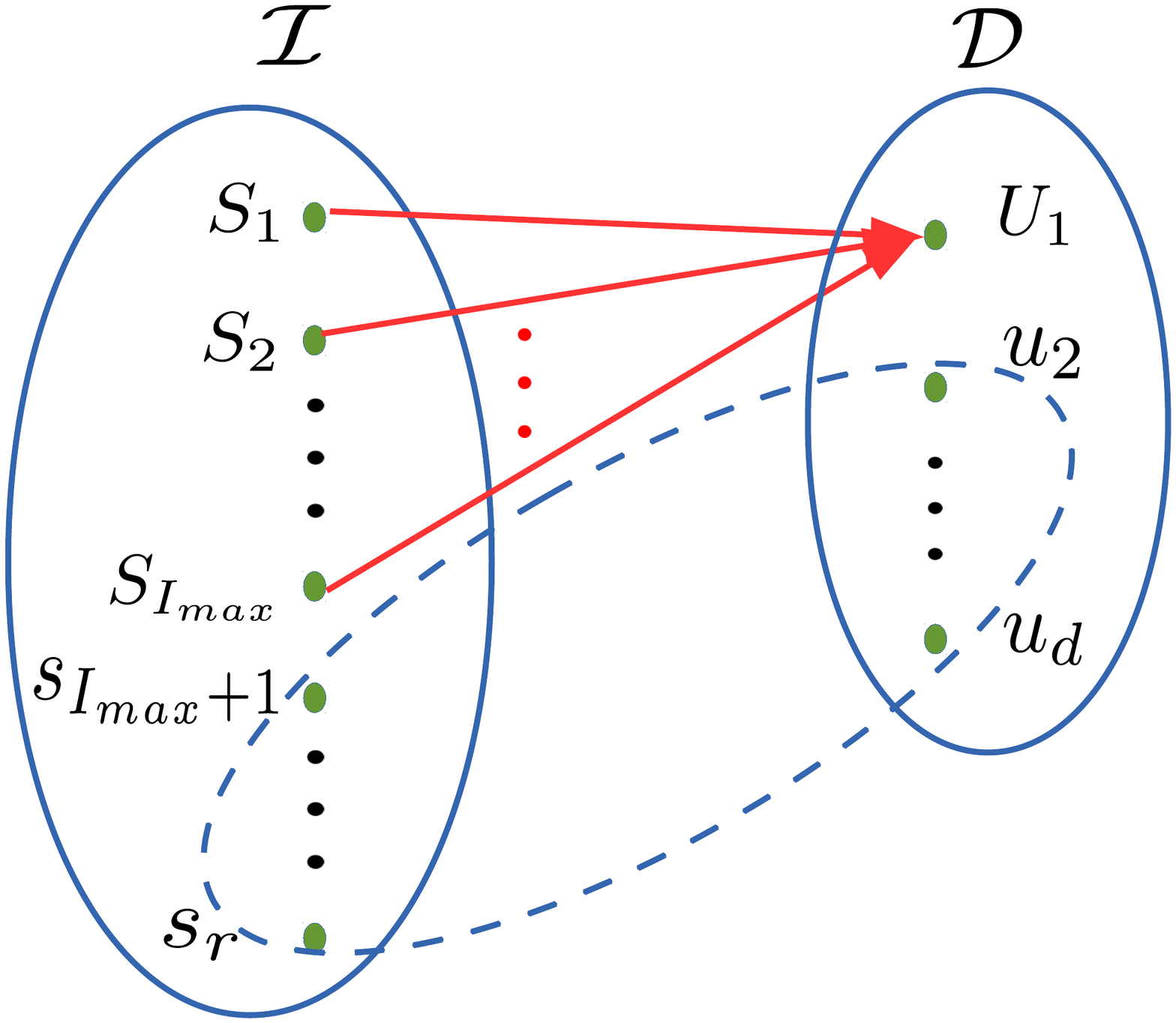}
\caption{A possible association pattern where, without loss of generality, $u_1$ is assumed to be a defective for which $|{\cal I}(u_1)|=I_{max}$. The set of inhibitors and defectives that are associated only among themselves (which the genie reveals) are inside the dotted ellipse.}
\label{fig-ass_LB_WSI}
\end{figure}
The second lower bound is obtained as shown below. There could exist at least one defective $u_1 \in {\cal D}$ so that $|{\cal I}(u_1)|=I_{max}$. Consider an association pattern where ${\cal I}(u_1) \cap {\cal I}(u_k) = \{\emptyset\}$, for $u_k \in {\cal D}, k\neq 1$, as depicted in Fig. \ref{fig-ass_LB_WSI}. Now, we use a similar argument as in the proof of the second lower bound in Theorem \ref{thm-LB_GTI_ID_NSI}. Let a genie reveal the inhibitor subset ${\cal I}-{\cal I}(u_1)$, the defective subset ${\cal D}-u_1$ and  their associations. Now, none of the items from the sets ${\cal I}-{\cal I}(u_1)$ and ${\cal D}-u_1$ is useful in distinguishing the inhibitors in the set ${\cal I}(u_1)$ from the unknown defective and the normal items. This is because the entropy of an outcome is zero if the test contains some defective from  ${\cal D}-u_1$ but none of its associated inhibitors (which are only from the set ${\cal I}-{\cal I}(u_1)$) as such a test outcome is always positive. The entropy of an outcome does not change if any of the inhibitors ${\cal I}-{\cal I}(u_1)$ with or without its associated defectives (which are only from the set ${\cal D}-u_1$) is present in the test. Thus, the problem is now reduced to the $1$-inhibitor problem of finding $I_{max}$ inhibitors amidst $n-(r-I_{max})-(d-1)$ normal items and one (unknown) defective. A lower bound on the number of non-adaptive tests for this problem is clearly a lower bound on the number of tests for the original problem of determining the association graph for the IDG-WSI model. Since $r,d=o(n)$ and $I_{max} \underset{n \rightarrow \infty}{\longrightarrow} \infty$, using Theorem \ref{thm-LB_1_Inh_Model}, we get the lower bound $\Omega\left(\frac{I^2_{max}}{\log I_{max}} \log n\right)$.

The third lower bound is obtained below for the case where $r= (c-1)d + kd$, for some constant $0<k<1$ and $I_{max}=c$. The proof for the other two cases mentioned in the statement of the theorem are similar. Parts of this proof are similar to the proof of the third lower bound in Theorem \ref{thm-LB_GTI_ID_NSI}, and hence we only point out the differences in this proof. As in the proof of Theorem \ref{thm-LB_GTI_ID_NSI}, we consider a reduced problem as follows. As depicted in Fig. \ref{fig-LB_WSI_d_sq}, a specific class of association graph is considered, where disjoint sets of $c-1$ inhibitors $\{{\cal I}_1, {\cal I}_2, \cdots, {\cal I}_d\}$ are associated with one defective each, i.e., each item in the set ${\cal I}_i$ with $|{\cal I}_i|=c-1$ is associated with the defective $U_i$, for $i=1,\cdots,d$. Each item in the set of inhibitors $ {\cal I}_{d+1}\triangleq\{S_{(c-1)d+1},\cdots,S_{(c-1)d+kd-1}\}$ is associated with one distinct defective with which the sets of inhibitors $\{{\cal I}_1,\cdots, {\cal I}_{kd-1}\}$ are also associated, i.e.,  $S_{(c-1)d+j}$ is associated with the defective $U_j$, for $j=1,\cdots,kd-1$. The remaining inhibitor $S_r$ is associated with exactly one of the defectives in the set ${\cal D}_{S_r}\triangleq\{U_{kd},\cdots,U_d\}$. It is now easily seen that the first constraint is satisfied, and $|{\cal I}(U_j)|\leq c$ for all $j$, which means that the second constraint is also satisfied. Now, let a genie reveal the realizations of $\{{\cal I}_1, {\cal I}_2, \cdots, {\cal I}_{d+1}\}$ and  $\{U_1,\cdots,U_{kd-1}\}$, given by  $\mathscr{I}\triangleq\{I_1, I_2, \cdots, I_{d+1}\}$ and $\overline{\mathscr{D}}_{S_r}\triangleq\{u_1,\cdots,u_{kd-1}\}$ respectively. The association pattern between them given by ${\cal E}(\mathscr{I},\overline{\mathscr{D}}_{S_r})$ is also revealed. 
\begin{figure*}
\includegraphics[totalheight=5.3in,width=6.5in]{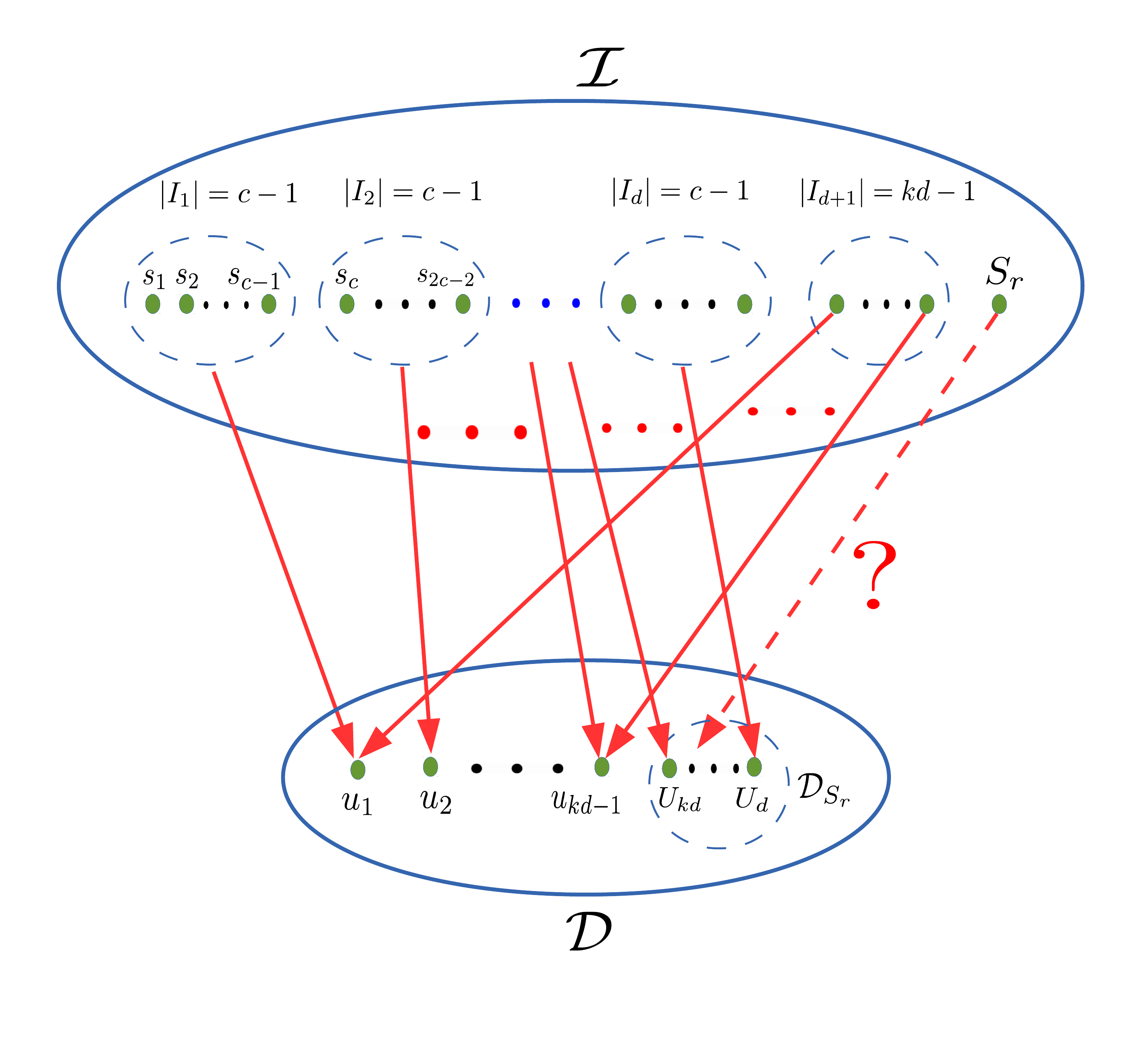}
\caption{Association graph with realizations $\{I_1,\cdots, I_d,I_{d+1}, u_1,\cdots,u_{kd-1}\}$ considered for obtaining the third lower bound for the IDG-WSI model. The genie reveals the realizations $\{I_1,\cdots, I_{kd-1},I_{d+1}\}$ along with their association pattern with the realizations $\{u_1,\cdots,u_{kd-1}\}$. It also reveals the realizations $\{I_{kd},\cdots,I_{d}\}$ which are known to be associated with the unknown realization of the remaining defectives ${\cal D}_{S_r}$. It is also known that the unknown inhibitor $S_r$ is associated with exactly one of the defectives in the set ${\cal D}_{S_r}$. Such an association graph is chosen so that the constraint $I_{max}=c$ is not violated.}
\label{fig-LB_WSI_d_sq}
\end{figure*}

The ``residual message'' in the system is now given by $W_1\triangleq \left\{S_r,{\cal D}_{S_r},{\cal E}(S_r,{\cal D}_{S_r}), {\cal E}(\{I_{kd},\cdots,I_d\},{\cal D}_{S_r})\right\}$. We now show that determining $W_1\triangleq{\cal E}(S_r,{\cal D}_{S_r})$ itself requires order of $d^2$ tests. It is easy to see that there is no reduction in the mutual information $I\left[W_1;\mathbf{y}|S_r,{\cal D}_{S_r},\mathscr{I},\overline{\mathscr{D}}_{S_r},{\cal E}(\mathscr{I},\overline{\mathscr{D}}_{S_r})\right]$ if the items in the set $\{I_1,\cdots,I_{kd-1},I_{d+1},\overline{\mathscr{D}}_{S_r}\}$ do not participate in any of the tests. So, we assume hereon that these items do not participate in any of the tests, and thus denote the rest of the items which participate in the tests by ${\cal W}'\triangleq {\cal N}\bigcup \{I_{kd},\cdots,I_d, S_r\} \bigcup {\cal D}_{S_r}$.

For the reduced problem considered, we have
\begin{align*}
&\underset{W_1}{\max}  ~\Pr\left\{\hat{W}_1\neq W_1\right\}\geq ~\mathbb{E}_{f}\left[\Pr\left\{\hat{W}_1\neq W_1\right\}\right]\triangleq P_{e_{avg}}.
\end{align*}where $f$ denotes some probability mass function of the residual association graph. Let $f_1$ and $f_2$ denote independent probability mass functions of the residual association patterns ${\cal E}(s_r,\mathscr{D}_{s_r})$ and ${\cal E}(\{I_{kd},\cdots,I_d\},\mathscr{D}_{s_r})$, for any realization of $(S_r,{\cal D}_{S_r})$ given by $(s_r,\mathscr{D}_{s_r})$. The function $f_2$ is such that 
\begin{align}\label{eqn-f_WSI}
\Pr\{(s_r,u)\in {\cal E}(s_r,\mathscr{D}_{s_r})\}=\frac{1}{d(1-k)+1}, \forall u \in \mathscr{D}_{s_r},
\end{align}for any realization $(s_r,\mathscr{D}_{s_r})$. Also, it is assumed $f$ is such that the realizations of $(S_r,{\cal D}_{S_r})$ are uniformly distributed across the rest of the items, i.e., occurrence of every realization happens with probability $\frac{1}{{n-r-kd+2 \choose d(1-k)+1}(n-r-d+1)}$.

Let $\mathbf{M}$ be the test matrix which is known a priori. Also, let the matrix $\mathbf{M}_1$ denote the test matrix $\mathbf{M}$ whose columns are  restricted to the items ${\cal W}'\backslash \{I_{kd},\cdots,I_d\}$, and the matrix $\mathbf{M}_2$ denotes the test matrix $\mathbf{M}$ whose columns are restricted to the items ${\cal W}'\backslash \{s_r\}$. Denote the ``virtual outcome vector'' obtained by testing the items using the matrices $\mathbf{M}_1$ and $\mathbf{M}_2$ by $\mathbf{y}_1\left({\cal E}(s_r,\mathscr{D}_{s_r})\right)$ and $\mathbf{y}_2\left({\cal E}(\{I_{kd},\cdots,I_d\},\mathscr{D}_{s_r})\right)$ respectively\footnote{The arguments of the virtual outcome vectors denote that the vectors are functions of their arguments.}. Note that $\mathbf{y}=\mathbf{y}_1.\mathbf{y}_2$, i.e., the actual outcome vector is equal to component-wise Boolean AND of the two virtual outcome vectors for every realization $(s_r,\mathscr{D}_{s_r})$. Since ${\cal E}(s_r,\mathscr{D}_{s_r})$ and ${\cal E}(\{I_{kd},\cdots,I_d\},\mathscr{D}_{s_r})$ are statistically independent messages, using data-processing inequality, we have 
\begin{align}\label{eqn-DP}
&I\left[{\cal E}(s_r,{\mathscr D}_{S_r});\mathbf{y}|s_r,{\mathscr D}_{s_r},\{I_{kd},\cdots,I_d\}\right]  \\\nonumber\leq ~& I\left[{\cal E}(s_r,{\mathscr D}_{S_r});\mathbf{y}_1|s_r,{\mathscr D}_{s_r}\right].
\end{align}
Now, applying Fano's inequality, we have
\begin{align}\nonumber
&H[{\cal E}(S_r,{\cal D}_{S_r})|S_r,{\cal D}_{S_r},\{I_{kd},\cdots,I_d\}]\\ \nonumber&= \frac{1}{{n-r-kd+2 \choose d(1-k)+1}(n-r-d+1)}\sum_{S_r,{\cal D}_{{\cal S}_r}}  \log (d(1-k)+1)\\
\nonumber
&\leq  1 + P_e H[{\cal E}(S_r,{\cal D}_{S_r})|S_r,{\cal D}_{S_r},\{I_{kd},\cdots,I_d\}] \\\nonumber&~~~~~~+ I\left[{\cal E}(S_r,{\cal D}_{S_r});\mathbf{y}|S_r,{\cal D}_{S_r},\{I_{kd},\cdots,I_d\}\right]
\end{align}
\begin{align}
\nonumber
&\leq  1 + P_e \log (d(1-k)+1)+\frac{1}{{n-r-kd+2 \choose d(1-k)+1}(n-r-d+1)}\times \\\label{eqn-MC}&~~~~~~~~\sum_{S_r,{\cal D}_{S_r}} I\left[{\cal E}(S_r,{\cal D}_{S_r});\mathbf{y}_1|\{S_r,{\cal D}_{S_r}\}=\{s_r,{\mathscr D}_{s_r}\}\right]\\
\nonumber
&\leq  1 + P_e \log (d(1-k)+1) +\frac{1}{{n-r-kd+2 \choose d(1-k)+1}(n-r-d+1)}\times \\\nonumber&~~~~~~~~~\sum_{S_r,{\cal D}_{S_r}} H\left[\mathbf{y}_1|\{S_r,{\cal D}_{S_r}\}=\{s_r,{\mathscr D}_{s_r}\}\right],
\end{align}where $P_e=\Pr \{\hat{{\cal E}}(S_r,{\cal D}_{S_r})\neq {\cal E}(S_r,{\cal D}_{S_r})\}\leq P_{e_{avg}}$ and  (\ref{eqn-MC}) follows from (\ref{eqn-DP}). Now, following similar steps after (\ref{eqn-cond_red_H}) in the proof of Theorem \ref{thm-LB_GTI_ID_NSI}, we have the lower bound of $\Omega((d(1-k)+1)^2)=\Omega(d^2)$ tests.
\end{proof}

Thus, in the $d=O(I_{max})$ and $d=O(r)$ regimes, the upper bound on the number of tests for the proposed non-adaptive pooling design is away from the proposed (second) lower bound for the IDG-WSI and IDG-NSI models by $\log I_{max}$ and $\log r$ multiplicative factors respectively. In the $I_{max}=o(d)$ and $r=o(d)$ regimes, the upper bounds exceed the proposed (third) lower bounds by $\log n$ multiplicative factors for both the IDG models, with some restrictions on $I_{max}$ or $r$ in IDG-WSI model. When these restrictions are removed, the evaluation of the lower bound might require consideration of other association graphs like in Fig. \ref{fig-Eg1}, as an extension of the association graph used in proof of the third lower bound in Theorem \ref{thm-LB_GTI_ID_WSI}. But even for the graph in Fig. \ref{fig-Eg1}, the optimization of the entropy over the pool size becomes combinatorially cumbersome. We thus relegate the evaluation of lower bound for the unconstrained IDG-WSI model to future work. For the proposed two-stage adaptive pooling design, the upper bound on the number of tests is away from the proposed (first) lower bound by $\log n$ multiplicative factors for both the IDG-WSI and IDG-NSI models in all regimes of the number of defectives and inhibitors.

\section{Conclusion}
A new generalization of the $1$-inhibitor model, termed IDG model was introduced. In the proposed model, an inhibitor can inhibit a non-empty subset of the defective set of items. Probabilistic non-adaptive pooling design and a two-stage adaptive pooling design were proposed and lower bounds on the number of tests were identified. Both in the small and large inhibitor regimes, the upper bound on the number of tests for the proposed non-adaptive pooling design is shown to be close to the lower bound, with a difference of logarithmic multiplicative factors in the number of items. 

For the proposed two-stage adaptive pooling design, the upper bound on the number of tests is close to the lower bound in all regimes of the number of inhibitors and defectives, the difference being logarithmic multiplicative factors in the number of items. 

Future works could include more practical versions of the IDG model, such as taking the following considerations into account.
\begin{enumerate}
\item Cancellation effect of the normal items on the inhibitors.
\item Partial inhibition of expression of defectives by the inhibitors, which also naturally embraces the presence of inhibitors in the semi-quantitative group testing model \cite{EmM_TIT2014}.
\item Inclusion of the $k$-inhibitor model, for unknown $k$, as a part of the association pattern in the IDG model.
\end{enumerate}
Obtaining lower and upper bounds on the number of tests for the aforementioned variants of the IDG model along with inclusion of noisy tests should be more involved and worth pursuing. 

\bibliographystyle{ieeetr}
\bibliography{References_GT}

\end{document}